\documentclass[12pt]{article}

 \pdfoutput=1 
\usepackage[applemac]{inputenc}
\usepackage[T1]{fontenc}
\usepackage{amsmath}
\usepackage{amsfonts}

\usepackage[english]{babel}

\usepackage{geometry}                
\usepackage[parfill]{parskip}    
\usepackage{graphicx}
\usepackage{amssymb}
\usepackage{epstopdf}
\usepackage[all]{xy}

\usepackage{amsthm} 
\newtheorem{thm}{Theorem}[section] 
\theoremstyle{definition} 
\newtheorem{dfn}{Definition}[section] 
\theoremstyle{axiom}

\theoremstyle{remark} 
 
\theoremstyle{plain} 

\newtheorem{prop}[thm]{Proposition}

\theoremstyle{plain}

\begin{document}

\title{{\it Pacotte's tree networks}, graph theory and projective geometry}
\author{Daniel Parrochia}
\date{University of Lyon (France)}
\maketitle

\textbf{Abstract}

The notion of "tree network" has sparked renewed interest in recent years, particularly in computer science and biology (neural network). However, this notion is usually interpreted in an extremely restrictive way: essentially linked to data processing, today's “tree networks” are hybrid network topologies in which star networks are generally interconnected via bus networks. These networks are, most often, hierarchical and regular, and each of their nodes can have an arbitrary number of child nodes. At the outset, however, the notion of "tree network", introduced in 1936 by Belgian physicist Julien Pacotte, was quite different: more general and, at the same time, more constrained, it should also serve an ambitious objective: the reconstruction of mathematics from concrete empirical structures. Usually poorly commented on and poorly understood (especially by philosophers), it had no real posterity. In this article, we first try to clarify this notion of "tree network" in the sense of Julien Pacotte, which makes it possible to eliminate the bad interpretations to which this notion has given rise. To this end, we use the language and concepts of graph theory and formalize the main properties of these networks which, contrary to popular belief, are not, in general, trees. In a second part, we then try to follow and explain, step by step, how Pacotte intended, using concepts borrowed from projective geometry, to reconstruct all of mathematics from such a network.

\textbf{Key words.}
Pacotte, Tree network, graph, bipartite graph, tree, cycle, cocycle, coloring of a graph, history of graph theory.

\section{Introduction}

	Physicist and epistemologist of Belgian origin, Julien Pacotte is a little-known thinker\footnote {It would seem to be Julien Désiré Humbert Ghislain Pacotte, born in La Louvière (Hainaut) in 1887 (see: La Louvière, Civil status, Births 1885 -1895, Act 116 digitized on Zoekakten.nl). A physicist by training, he is said to have been an assistant at the Royal Observatory of Belgium (Archives of the Royal Observatory, personal file, box 029, no. 976), then a researcher at the National Fund for Scientific Research in Brussels. Relatives say he may have died in 1956 or 1957 at the Tombeek sanatorium (Overijse), the consequences of tuberculosis. His grave is in the Nivelles cemetery. Pacotte's father, Emile, was of French origin.}. His work (especially on physics and technical thinking) has been cited, in its time, by French authors like Gaston Bachelard (see \cite{Bac}, 250) and Georges Canguilhem (see \cite{Can}, 131) but, as recently observed François Sigaut (see \cite{Sig}, it is one of the great "oversights" in the history of science and technology. We are interested here in one of his undoubtedly deepest book - {\it Le réseau arborescent, schème primordial de la pensée (The Tree Network, primordial scheme of thought)} (1936) - of which there are very few mentions in the literature\footnote{The first reference, signed by Jean Ullmo, and which is not a review, is limited to mentioning the existence of the book (see \cite{Ull}, 376); the second (see \cite{Bou}, 124) only indicates that “starting from the idea that the 'image of a current dividing into ramifications is the primordial schema of thought, the author tries to build a geometry of the tree structure with the aim of shedding light on the bases of formal thought”. We can also mention a review of E. Pinte (see \cite{Pin}, 14)  and another one from P. Schrecker, the translator of Leibniz (see \cite{Sch}, this 24-line text being by far the most faithful one can find. The author, who does not go into detail on Pacotte's constructions, considers his conception "simple and fruitful". He evokes about it the generation of numbers according to Plato, "as it was recently renewed by Mr. Brouwer".}. This work was commented on by Gilles Deleuze and Félix Guattari, who made it the spearhead of classical thought, fundamentally linked, according to them, to the notion of "tree" (see \cite{Del}, 25)\footnote{Unfortunately, the authors have superficially read the beginning of the book  because the diagram in note 12, supposed to reproduce an example of a tree network, is false.} - to which they claim to oppose that of {\it rhizome}\footnote {Obviously, the authors lack mathematical knowledge: there is no need to think very long to realize that the rhizome is also a tree, at least in the mathematical sense of the term.}. We will show that there is a misinterpretation here, because Pacotte's networks are not necessarily trees and assume complex mathematics, far from the usual trivialities. In reality, Julien Pacotte was one of the best experts one can find in the science of his time\footnote{Evidenced, among others, the list of publications by Julien Pacotte between 1920 and 1940: \cite{Pac1}, \cite{Pac2}, \cite{Pac3}, \cite{Pac4}, \cite{Pac5}, \cite{Pac6}, \cite{Pac7}, \cite{Pac8}, \cite{Pac9}, \cite{Pac10}, \cite{Pac11}, \cite{Pac12} and its competence is incommensurate with that of his critics.}. 
	Unfortunately, this skill, as well as the novelty of its perspective, does not seem to have been perceived more abroad than in France. Christy Wampole, for example, refers to Pacotte's work in a book devoted to the metaphor of the tree, but only to indicate that it was taken by the author as a metaphor for thought (see \cite{Wam}, 183). Same observation for the article by Matteo Pasquinelli, where the author, once again, quotes the text of Pacotte in support of a metaphorical conception of the tree (see \cite{Pas}), which is obviously , as in the previous cases - Deleuze-Guattari understood -, a total error, since it is on the contrary for the author a "formal scheme" with its laws, not a simple image. In the following, we will first redefine the Pacotte network in terms of graph. In a second step, we will comment on its application and its raison d'être, on which the previous authors are silent (the tree reconstruction of mathematics). \\

\Large
\qquad\qquad Partie I : Nature and properties of the network
\normalsize

\section{The tree network}

Tree networks, today, have structures like the following (see Fig. 1):
	
\begin {figure} [h]
\centering
\vspace {-2 \baselineskip}
\includegraphics [width = 6in] {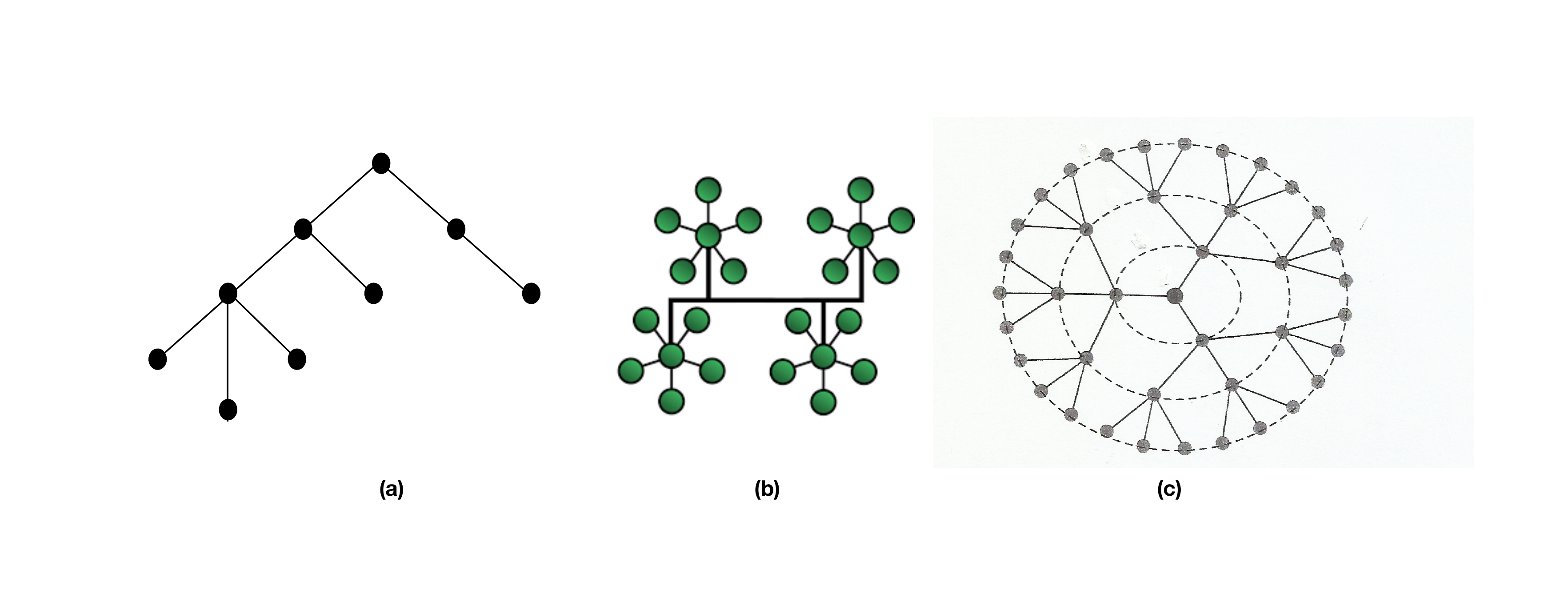}
\vspace {-2 \baselineskip}
\caption {Example of tree network structures}
\label {fig: Elements}
\end{figure}

In the literature (see \cite{Bra}; \cite{Sos}), we meet essentially trees, i.e connected graphs without cycles (see Fig.1 (a)), star-bus networks, i.e hybrid network topologies in which star networks are interconnected via bus networks (see Fig.1 (b)), and regular hierarchical $n$-graphs (see  \cite{Kro}) as the hierarchical $3$-graph (see Fig.1 (c)) – which could have described in the past the structure of the Hegelian Encyclopedy.

In fact, the tree networks, as introduced by Julien Pacotte, were quite different, in the sense that they integrated into a single structure tree and circular aspects.

“We call tree structure, wrote Julien Pacotte at the beginning of this work, the system of line segments representing a current which splits, thus generating currents which can split in their turn and so on. Thus, in ordinary space, the figure formed by a plant, from the soil to the birth of the leaves, or even, the cell which splits and multiplies by successive generations, or even the colloidal grain which disintegrates, the molecule which dissociates, the atom which ionizes, the atomic nucleus which disintegrates in stages, etc." (see \cite{Pac8}, 3).

It will be noted, with the author, that such oriented "trees" are divergent. But there is no doubt that we can also find examples where the orientation of the trees is reversed:

"Thus the crystallization of a heterogeneous magma in the middle of a granite structure, the coalescence of the micelles, the march of the half-chromosomes towards one of the two poles of the ovoid nucleus to form a new nucleus, are, in the space-time continuum, converging trees"\cite{Pac8}, 4).

On the basis of this intuitive data, Pacotte introduced the more complex idea of a tree-like, circulatory and polarized "network": nature, as well as life, seems in fact to attach great importance to convergent-divergent tree systems (plants, circulatory system, nervous system, etc.)\footnote{Pacotte seems to have noticed that natural networks can be both tree-like {\it and} looped, which has only been highlighted fairly recently for plant networks. See \cite{Kat}; \cite{Cor}.}. In general, the generation of collections seems to be reduced to the schema of the tree structure, so that, if we admit that pure mathematics and formal logic have as their sole object collections and numbers, then "an original theory of polarized networks would represent the systematic exposition of the foundations of formal, logical and mathematical thought" (see \cite{Pac8}, 5).

Julien Pacotte, at a time when – let us remember – the main concepts of order theory and graph theory did not yet exist in all their clarity, therefore set out to build a "geometry of ramifications, an entirely new doctrine, independent like amorphous geometry\footnote{This means undoubtedly topology.} or pure projective geometry, but valid as the foundation of formal logic and pure mathematics” (see \cite{Pac8}, 6).

We will try to reread first it in the light of the concepts and structures of graph theory.

\section {Basic concepts: the idea of a "polarized network"}

Given segments of lines radiating around a point and traversed, each in a determined direction, we introduce, according to Pacotte, the following definitions. \\

\begin{dfn}
We call {\it diffluence} a set formed by $ n $ divergent arcs of a vertex with 1 single arc converging towards it. \\
\end{dfn}

\begin{dfn}
We call {\it confluence} a set formed by $ p $ arcs converging towards a vertex with 1 single diverging arc oriented outwards. \\
\end{dfn}

\begin{dfn}
In a diffluence (resp. confluence), the system of divergent (resp. convergent) currents is called a {\it sheaf}\footnote{This structure should not be confused with the mathematical structure of {\it sheaf} ("faisceau", in French).The French word, used by Pacotte here, is "gerbe".}.
In both cases, the remaining single segment is named {\it opposite segment}. \\
\end{dfn}

\subsection{Visualisation}

	Let two successive diffluences (resp. confluences). We can assume that a ray of the sheaf of one is the opposite segment of the other. Furthermore, the same line can serve as the opposite segment for a confluence and a diffluence. Let us first visualize this in the case of two diffluences (Fig. 2 left) and in that of a confluence and an associated diffluence (Fig. 2 right). \\

\begin {figure} [h]
\centering
\vspace {-3 \baselineskip}
\includegraphics [width = 6in] {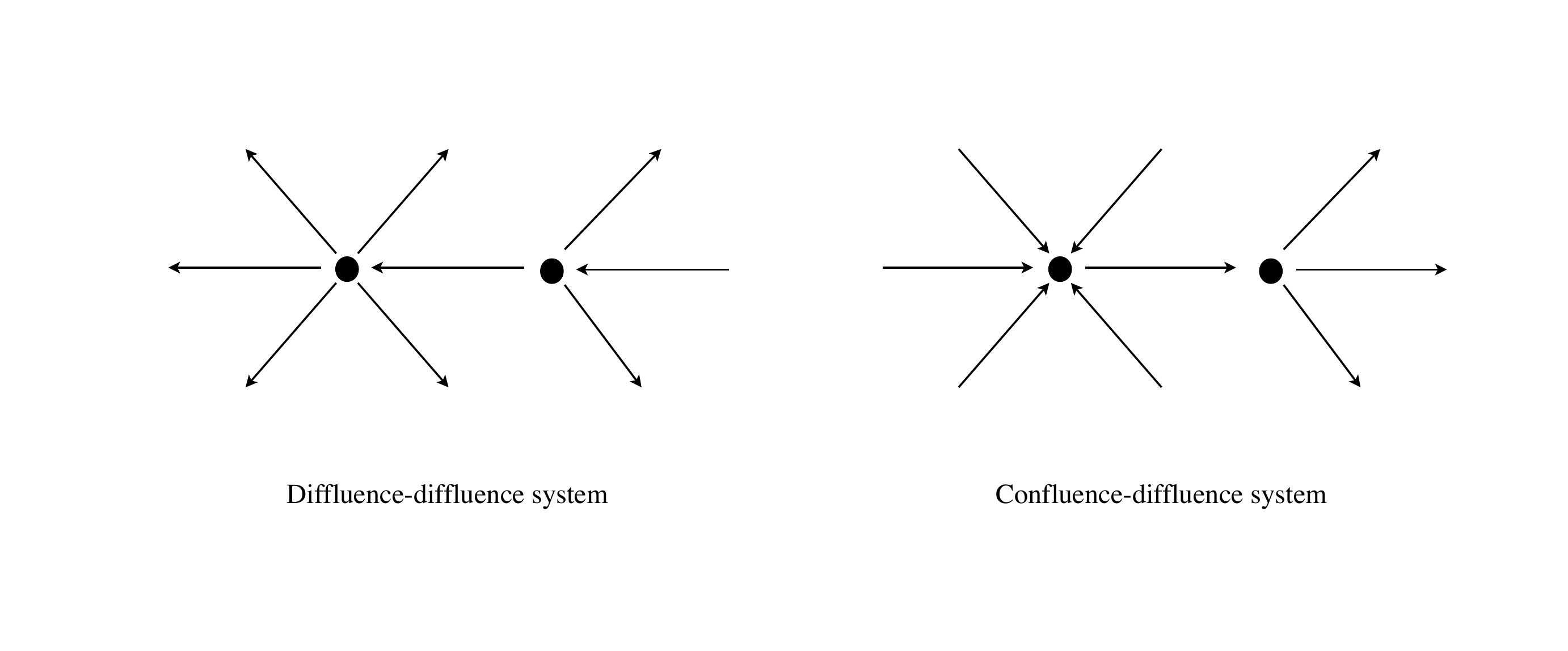}
\vspace {-3 \baselineskip}
\caption {The elements of the polarized network}
\label {fig: Elements1}
\end{figure}

\begin{dfn}
A network of oriented segments assembled by their ends and forming confluences and diffluences is called a {\it polarized network}.
\end{dfn}

\subsection {Properties of the polarized network}

The polarized network has the following properties:

\begin{enumerate}
\item There is a current everywhere in the network, the direction of which is permanent along a segment;
\item It is permissible to return all the currents at once: a diffluence becomes a confluence and vice versa;
\item Any network accepts a {\it direct current} and a {\it reverse current}. We will name {\it positive poles} the diffluences of the direct current;
\item The sign of the poles results from the direction of the current.
\end{enumerate}
	
\subsection {The notion of {\it dipole}}

Let us take again the case of two successive diffluences: the line joining the two positive poles plays different roles with respect to the two poles, being on one side a ray of sheaf, on the other an opposite segment. So that each segment has its own role, the same for its two poles, Pacotte cuts the line between the two positive poles by inserting a new point. In the latter, there is neither confluence nor diffluence, but simply {\it fluence}. This point will also be called a pole, but it will be said {\it negative}. Conversely, between two negative poles, a positive pole will always be inserted. In the case of a consecutive diffluence and confluence, the common line will only play the role of opposite segment and there will be no need to insert a pole of fluence (see Fig. 3).

\begin{figure} [h]
\centering
\vspace {-2 \baselineskip}
\includegraphics [width = 4in] {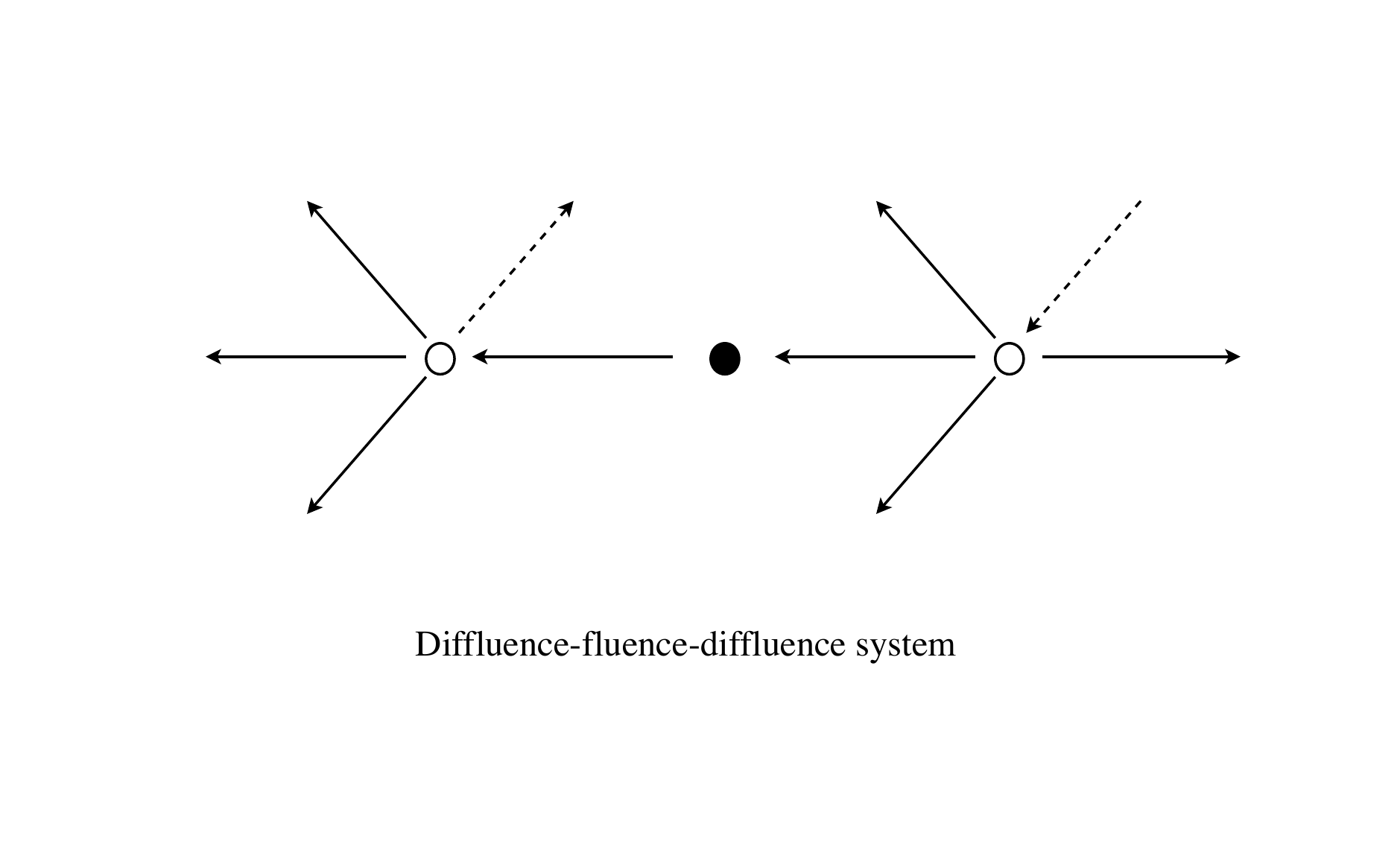}
\vspace {-3 \baselineskip}
\caption {The dipole network}
\label {fig: Elements2}
\end{figure}

Under these conditions, the network has the following new properties, which are deduced from the above: \\

\begin{prop}
The negative and positive poles alternate everywhere. \\
\end{prop}
\begin {prop}
The segments of the network are of two kinds: some are sheaves, others are opposite segments. \\
\end {prop}
\begin {prop}
A suitable sign being assigned to the free ends, all the segments are {\it dipoles}. \\
\end {prop}
\begin {dfn}
We will call {\it unit-dipoles} the sheaf rays and {\it connection dipoles} the opposite segments.
\end {dfn}

Knowing that one has defined as {\it positive} the pole which, for the direct current, is a diffluence, the conditions of the alternation make it possible to deduce further: \\
 
 \begin {prop}
 In a unit dipole, the current goes from the positive pole to the negative pole. \\
 \end {prop}
 \begin {prop}
 In a connecting dipole, the current goes from the negative pole to the positive pole.
 \end {prop}

 We see that the permanence of the current's direction along a line formed by several dipoles is expressed by the alternation of the dipoles of the two kinds. We deduce that the use of alternating poles makes it possible to put in the background the ideas of current, of direct direction and of opposite direction. Therefore, the sign of the poles (positive or negative)  now makes the notation of the current's direction unnecessary. \\
 
 \begin {dfn}
 Today, we call {\it vertex coloring} the assignment of a color to each of the vertices of a graph $G(X, U)$, with $X$, set of vertices, $U$, set of arcs or edges, so that two adjacent vertices do not receive the same color. \\
 \end {dfn} 
 
 \begin {dfn}
 A graph is said to be {\it p-chromatic} if its vertices admit coloring in $p$ colors. We will call {\it chromatic number} $\gamma(G)$ of a graph $G$ the minimum number of distinct colors needed to color its vertices.
 \end {dfn}

 We deduce the following proposition: \\
 
 \begin {prop}
 A Pacotte tree network is a bi-chromatic graph. Its chromatic number $\gamma (G) = 2$.
 \end {prop}
 
 In other words, we immediately deduce from the extremely general conditions previously posed the existence of a bicolouring of the network by taking, for example, positive poles as white points and negative poles as black points. As Pacotte anticipated, such a bicolouring makes it possible to get rid of any indication of orientation in the network which thus gains in generality.

On the graph, Pacotte intends to distinguish the connection dipoles from the unit dipoles in the following way: in a connection dipole, the line of the poles will effectively reach the two poles. It will neither reach them in a unit dipole. This leads to the following configurations (see Fig. 4).

\begin {figure} [h]
\centering
\vspace {0 \baselineskip}
\includegraphics [width = 2.5in] {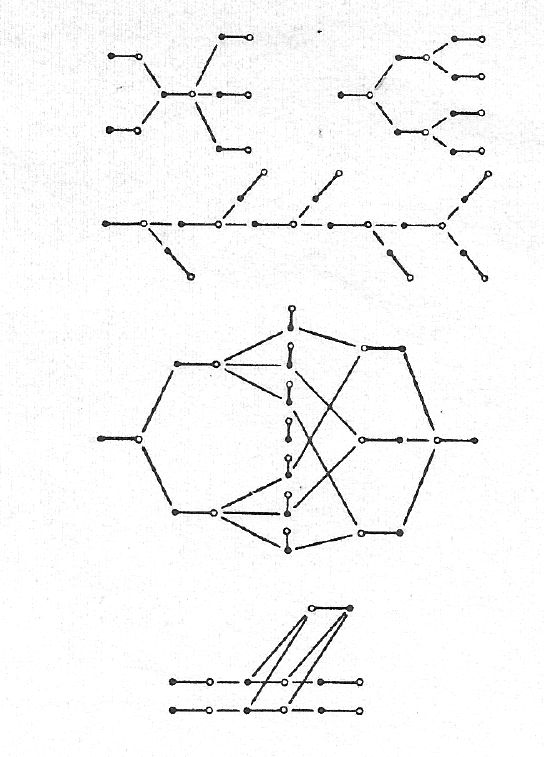}
\vspace {-1 \baselineskip}
\caption {Examples of Pacotte networks (from \cite{Pac8}, 9)}
\label {fig: Elements3}
\end {figure}

  For clarity, however, we will restore the direction of the current in the 4$^{th} $ diagram and number the vertices. We thus obtain the network of Fig. 5, which will further serve as a reference network. This figure perfectly illustrates the notion of "circulatory network" put forward by Pacotte at the beginning of his work (see \cite{Pac8}, 4).
  
  \begin {figure} [h]
\centering
\vspace {0 \baselineskip}
\includegraphics [width = 4in] {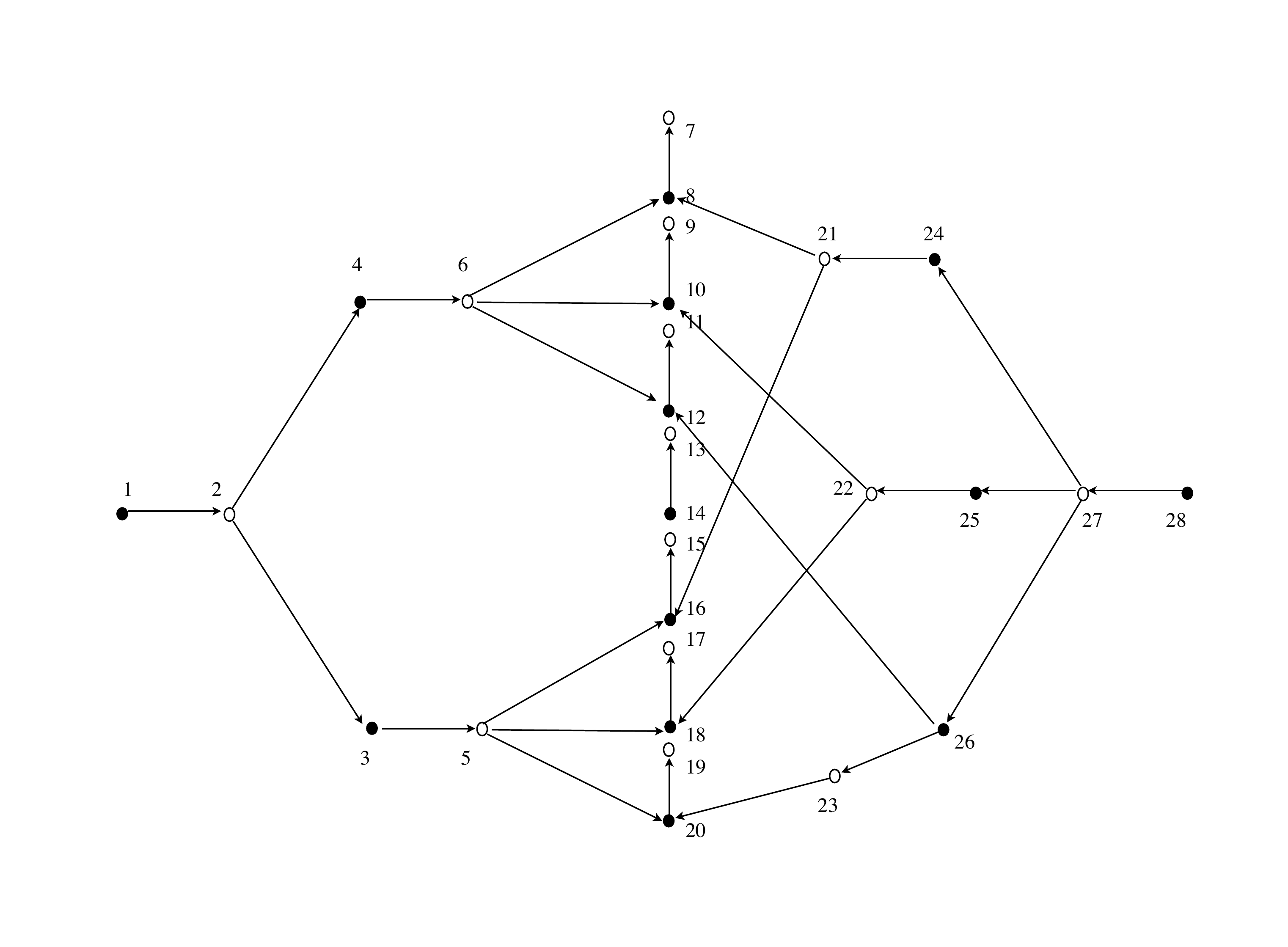}
\vspace {-1 \baselineskip}
\caption {Pacotte network with currents}
\label {fig: Currents4}
\end {figure}

   The convention of the alternation of poles however leads to a more advanced schematization which is the alternation of the dipoles. As Pacotte shows, an alternating chain of these entities thus presents 4 types of objects: 2 types of poles and 2 types of segments. On the set $X = \{1, 2, 3, ... \}$ of the vertices of the network, we can then play two types of operators: the operator of cyclic permutation of the points and that of the reversal of the current's direction. The combination of the two creates the dihedral group.
   
 \subsection {Welds, sheaves and chains}
 
   If we recapitulate, the polarized network is therefore presented as follows (see \cite{Pac8}, 11) :
   
   \begin {enumerate}
   
   \item The network has two types of poles (positive and negative) and two types of dipoles (unit dipoles and connecting dipoles). We admit that the poles always belong to the dipoles and we call {\it welded joint} or, more briefly, {\it weld} ("soudure" in French), the system of two or more poles of the same sign.
   \item A {\it sheaf} is a system of unit dipoles assembled by a single weld; it is said to be positive or negative depending on the sign thereof.
   \item A {\it chain} is a system of arbitrary dipoles, the welds of which join two dipoles. A chain of unit dipoles has alternately positive and negative sheaves (two unit dipoles). A chain in which dipoles of the two types alternate (or {\it alternated chain}) has no sheaves.
   \item A polarized network is a system of dipoles, characterized as follows: each connecting dipole can attach, in each of its poles, one or more unit dipoles, but no connecting dipole; each unit dipole attaches, in each of its poles, a connecting dipole and only one. There is no connecting dipole of which no pole would be a weld.
   \end {enumerate}

Here is an illustration of these concepts from the diagram in Fig. 5.
 
-- The pairs of vertices (1, 3) or (1, 4), (2, 4) or (2, 5) are {\it welds}.

- The system of unit dipoles (2-3), (5-16), (5-18), (5-20), assembled by a single weld, the positive weld (2-5), is a {\it positive sheaf}.

- The {\it chain} (28-27, 27-24, 24-21, 21-8, 8-7) is a system of arbitrary dipoles, comprising both unit dipoles (27-24, 21-8 ) and connecting dipoles (28-27, 24-21, 8-7).

- The chain (23-20, 20-5, 5-16, 16-21, 21-8) is a {\it chain of unit dipoles}. It has alternately positive (23-20, 5-16, 21-8) and negative (20-5, 16-21) sheaves.

- The chain (28-27, 27-25, 25-22, 22-10), where alternate connecting dipoles and unit dipoles, does not have a sheaf because its system of unit dipoles (27-25, 22-10) is only joined by a single connecting dipole (22-25) comprising only a single positive pole (22) (resp. a single negative pole (25)). This is not enough to make a positive (resp. negative) sheaf since two poles with the same sign are required for this.

\subsection {A graph foundation for formal thinking}

In Pacotte's mind, the previous formalization is not a system of conventions among others but "an inevitable formal object, originally suitable for the schematization of qualitative ramifications which are a universal aspect of intimate reality" (see \cite{Pac8}, 11). In other words, the previous axiomatic would present an "irreducible evidence" which is in fact a true foundation of formal, logical and mathematical thought. The polarized network is therefore not a structure which would fit into it or be deduced from it. It's the opposite. Logical-mathematical structures, from the simplest to the most complex ones, must be deduced from such a network. Thus "any collection, any number is a sheaf"(see \cite{Pac8}, 11).

Despite these assertions – which must, once again, be understood in the context of a time when neither the theory of graphs, nor that of partial orders are yet fully developed – we will continue to reinterpret Pacotte's network in light of modern mathematical concepts and structures.

\section {Domain, layer, concordance}

In chapter 2 of his work (see \cite{Pac8}, 12 sq.), Pacotte comes to introduce the following new concepts:

\subsection {Domains and cocycles}

\begin {dfn}
A region of the network which, from a pole, is traversed by a current of the same direction is called the {\it domain} of the pole.
\end {dfn}
 
Given a subset $A \subset X $ of vertices of the graph $G (X, U)$, the domain of $ A = \{s \} $, where $ A $ is reduced to singleton $ \{s \} $ and which we will more simply denote dom($\{s \}$), therefore includes currents of the same direction coming from $ \{s \} $, that is to say, generally, and for any network, "first a sheaf, then a tree"\ (see \cite{Pac8},12).

This notion of {\it domain} is not obvious to translate into graph theory. For clarity, let us first recall here some definitions familiar to graph theorists. \\

\begin {dfn} [cocycles]
For any region $A$ of a graph $ G (X, U) $, we define:

- $ \omega^+ (A) $: the set of arcs having their initial end in $ A $ and their terminal end in $ \bar {A} = X - A $;

- $ \omega^ - (A) $: the set of arcs having their terminal end in $ A $ and their initial end in $ \bar {A} = X - A $;

We note: $\omega (A) = \omega^ + (A) + \omega^ - (A) $, the set of arcs or edges called  {\it cocycle} of the graph. 
\end {dfn}

These definitions being recalled, we could express today Pacotte's conception as follows. Starting from a connected oriented 1-graph $ G (X, U) $, his approach amounts to fixing certain conditions on the {\it cocycles}. Suppose that $ A = \{s \} $ has only one vertex. It will be said that:

\begin {enumerate}
\item $ \omega ^ + (A) $ is a {\it diffluence} if $ | \omega^+ (A) | = | \omega (A) -1 | $ (in other words, if $ | \omega ^ - (A) | = 1) $;
\item $ \omega ^ - (A) $ is a {\it confluence} if $ | \omega^- (A) | = | \omega (A) -1 | $ (in other words, if $ | \omega ^ + (A) | = 1) $;
\item $ \omega (A) $ is a {\it fluence} if $ | \omega^+ (A) | = | \omega ^ - (A) | = 1 $ (in other words, if $ | \omega (A) | $ = 2).
\end {enumerate}

Thus, in the graph of Fig. 4, for $ A = \{2 \} $, we have $ \omega^ + (\{2 \}) = \{4 \} $. Then put $ B = \{6 \} $, we now have $ \omega^+ (\{6 \}) = \{8, 10, 12 \} $. Therefore, dom($ \{2 \}) = \omega^+ (\{2 \}) \cup \omega^+ (\{6 \}) = \{4, 8, 10, 12 \} = \omega^+ (A, B) $. \\

\begin {dfn}
In the language of Pacotte, a network which presents only diffluences (resp. confluences) is a {\it tree}.
\end {dfn}

This definition, which amounts to defining a tree (or a rooted tree) by its cocycles, is dual from the definitions of graph theory. Indeed, in graph theory, we call {\it tree} a connected graph without cycles and {\it rooted tree} (graph $ \mathcal {F} = (X, T) $), a tree with root $ r \in X $ and such that, for any vertex $ j \in X $, there exists in $T$ a path going from $ r $ to $ j $. This does not mean that every Pacotte network is a tree (or a rooted tree). Some may actually have {\it cycles}. \\

\begin {dfn}
In a graph $ G (X, U) $, a {\it cycle} is a sequence of arcs:
\[
\mu = (u_ {1}, u_ {2}, ..., u_ {q}),
\]
such that:

(1) any arc $ u_{k} $ (with $ 1 <k <q $) is connected by one of its ends to the previous $ u_{k-1} $ and by the other end to $ u_{k + 1 } $ (in other words, it's a chain);

(2) the sequence does not use the same arc twice;

(3) the initial vertex and the terminal vertex of the chain coincide.

An {\it elementary} cycle also checks:

(4) in traversing the cycle, one meets the same vertex only once (except of course the initial vertex which coincides with the final vertex).
\end {dfn}

We then arrive at the following theorem, which shows the gross error made by Deleuze and Guattari: \\

\begin {thm}
A Pacotte tree network contains a cycle if one of its connected components of cardinal $ n $ contains $ n $ edges or more.
\end {thm}

\begin {proof}
Let $ G $ be a Pacotte network, $ C $ a connected component of $ G $ of cardinal $ n $ (i.e. having $ n $ vertices) and $ A $ a tree covering for this component. $ A $ has $ n-1 $ edges. If $ C $ has $ n $ edges, then there is $ a = (s, t) \in C $, which is not in $ A $. There is therefore a path $ c $, from $ s $ to $ t $, in $ A $ which does not contain $ a $. And, therefore, $ c $ {\it plus} $ a $ forms a cycle of $ C $. To find out if $ G $ has cycles, we therefore consider each component. If one of them contains $ n $ or more edges, then $ G $ contains a cycle.

\begin {figure} [h]
\centering
\vspace {-2 \baselineskip}
\includegraphics [width = 4in] {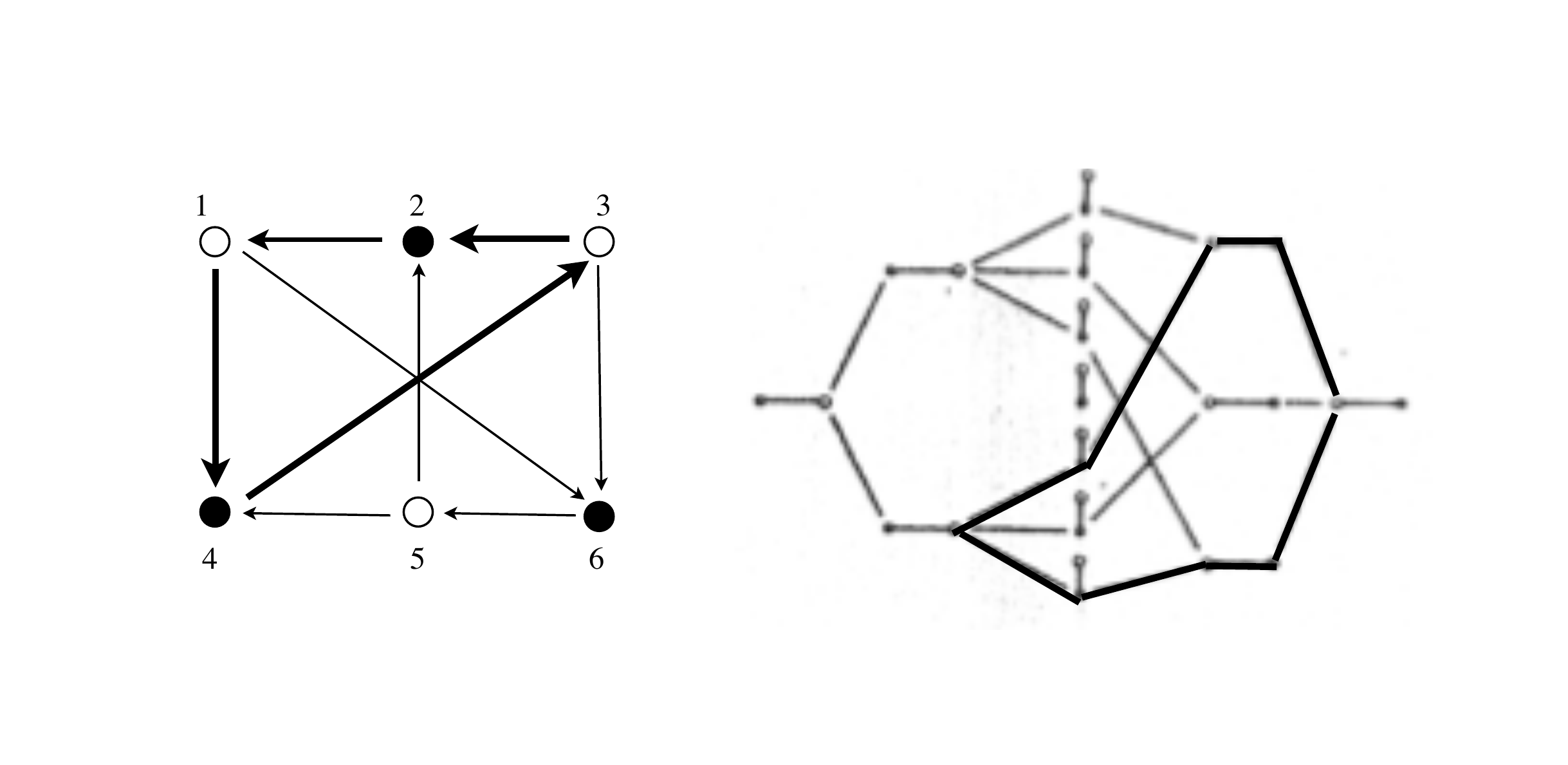}
\vspace {-2 \baselineskip}
\caption {Even length cycles in a Pacotte network}
\label {fig: Cycle5}
\end {figure}

In Fig. 6, examples of cycles are shown in bold. This makes it clear that a Pacotte network is not necessarily a tree or a rooted tree.
\end {proof}

Now remember the following definition: \\

\begin {dfn}
We call {\it circuit} a cycle $ \mu = (u_ {1}, u_ {2}, ..., u_ {q}) $ such that, for all $ i <q $, the terminal end of $ u_ {i} $ coincides with the initial end of $ u_ {i + 1} $. \\
\end {dfn}

\begin{thm}
A necessary and sufficient condition for a Pacotte network $ G (X, U) $ to be circuit-free is that
any non-empty subset of vertices $ A \subset X $ admits at least one element whose
all the predecessors are in $ \bar{A} $, the complement of $ A $. In other words, the subgraph
$ G_{A} $ has at least one vertex without predecessor.
\end {thm}

\begin{proof}
Suppose that $ G (X, U) $ has a circuit [$ x_ {0}, x_ {1}, ..., x_ {r} $]. Let $ A = \{x_ {0}, x_ {1}, ..., x_ {r} \} $ and
consider the subgraph $ G_ {A} $: any vertex of $ G_ {A} $ belonging to circuit has at
least a predecessor in $ A $. There is therefore a set $ A $ for which the
property is false. Conversely suppose $ G (X, U) $ without circuit and the property
false. There is therefore a graph $ G_ {A} $ whose all vertices have at least one
predecessor in $ A $. Let's start from $ x_ {i_ {0}} $ in $ G_ {A} $. $ x_ {i_ {0}} $ has a predecessor $ x_ {i_ {1}} $, ..., $ x_ {i_ {n-1}} $ has a predecessor $ x_ {i_ {n}} $. We can therefore build a path [$ x_ {i_ {n}}, x_ {i_ {n-1}}, ..., x_ {i_ {0}} $]. Since the graph has $ n $ distinct vertices, two vertices of this path are identical. This graph
therefore has a circuit, which is absurd.
\end{proof}

In the general case, a Pacotte network can however admit circuits. This is not the case with the network of Fig. 5, but this is the case of the network of Fig. 6 (left) which admits several circuits: (3, 2, 1, 4, 3), (3, 6, 5, 4, 3), (3, 6, 5, 2, 1, 4, 3), ( 3, 2, 1, 6, 5, 4, 3). \\

\begin {thm}
When a Pacotte network admits cycles or circuits, these are necessarily of even length.
\end {thm}

\begin {proof}
Pacotte's network is a bipartite graph. Now a bipartite graph does not have cycles of odd length.
\end {proof}

Let us now take the point of view, no longer of currents, but of dipoles. To traverse the domain by following the direction of the current amounts going from the connecting dipole associated with the initial pole to its unit dipole. We admit with Pacotte (see \cite{Pac8}, 13) that the connecting dipole, which belongs to the chosen pole, already belongs to the domain. This is the {\it initial dipole}. If the chosen pole is a free pole (which is not a weld), then it ends with a connecting dipole and, in this case, the domain is reduced to this very dipole. Otherwise, the connecting dipole is linked to other connecting dipoles, each time through a unit dipole.

\subsection {The notion of {\it layer}}

Pacotte then turns his attention to the moving front of the domain. “Since it is necessary to take, in each pole with the same sign as the given pole, finally the whole sheaf, suppose that we take it each time at once; and since it is necessary to continue each line initiated by the sheaf, up to the pole with the same sign as its center, suppose that, in one fell swoop, we take the sheaf and the connecting dipoles which immediately extend it. Under these conditions, the domain advances by well-defined leaps; however, successive front lines are not imposed. The first leap is imposed; the second will start from any line of the sheaf; the third will start from another line of the first sheaf or from any line of the sheaf generated on the second leap; And so on. We call {\it layer} of the domain, a system of poles which defines a possible front of the domain. The layer is made up of poles with the same sign as the domain pole: it has the sign of this pole” (see \cite{Pac8}, 13).

To make sense of this notion of {\it layer} from the point of view of graph theory, new definitions should be introduced. \\

\begin {dfn}
Given a graph $ G (X, U) $, we associate with each arc $ u \in U $ a number noted $ \ell (u) \in \mathbb {R} $ and called "length of the arc". We say that $ G $ is {\it valued} by the lengths $ \ell (u) $. If $ u = (i, j) $ we can also use the notation $ \ell_{ij} $ for the length of the arc $ u $.
\end {dfn}

In the case of the Pacotte network, we can associate with any dipole $ (i, j) $ the same length $ \ell_ {ij} = 1 $.

 We will now consider a Pacotte network $ R $ without circuits, provided with a root 1. Under these conditions, we can introduce the following {\it rank} function: \\

\begin {dfn}
The {\it rank} function associated with a graph without circuit of root 1 is obtained by associating, with each vertex $ i \in X $ of the graph, a positive integer $ r (i) $ called {\it rank} of the vertex $ i $ such as:
\[
r (1) = 0,
\]
$ r (i) $ being the number of arcs in a path of maximum cardinality between 1 and $ i $.
\end {dfn}

We can then decompose the graph into levels, each level corresponding to the vertices having the same rank. In the case of the Pacotte network, the poles of the same rank are called {\it normal layers}. As Pacotte explains, they form an ordered whole. However, between any two layers of a domain, even a tree, there is, in general, no order relation. We therefore arrive at the following proposition: \\

\begin {prop} [Pacotte]
The set of layers of the Pacotte network is not, in general, a total order.
\end {prop}

It is therefore only on the set of normal layers that we can establish an order relation. \\

\begin {dfn}
Let us call $\Gamma $ the multivalued mapping which, with any element of $ X $ associates a part of $ X $ (i.e. an element of $ P (X) $). \\
\end {dfn}

\begin {thm} [\cite{Gon}, 44]
In any circuit-free graph, there exists at least a vertex $ i $ such as $ \Gamma^ {- 1} _ {i} = \emptyset $. \\
\end {thm}

\begin {dfn}
We call $ {\it antibase} $ of any graph a set $ A \subset X $ satisfying the following two conditions:

(1) $ \forall j \in \bar{A}, \exists i \in A $ such that there is a path from $ j $ to $ i $;

(2) There are no paths between two vertices of the antibase.
\end {dfn}

In a circuit-free graph, the set of vertices having no successor forms the only antibase. If the antibase is reduced to a single element, it will be called the {\it antiroot} of the graph.

In the Pacotte network of Fig. 4, vertices 7, 9, 11, 13, 15, 17, 19 and 28 thus form the sole antibase of the graph.

One also proves the following theorems: \\

\begin {thm} [\cite{Ber}]
If $ G $ is a circuit-free graph with at least one arc, by any arc passes a cocircuit. \\
\end {thm}

\begin {thm} [\cite{Ber}]
A graph is strongly connected if and only if by any arc passes a circuit.
\end {thm}

We deduce that a Pacotte network without circuits cannot be strongly connected. However, we will prove later that it is, in this case, {\it quasi-strongly connected}.

\subsection {The notion of {\it concordance}}

Pacotte names {\it concordant} the vertices having a common layer (see \cite{Pac8}, 13-14) and which will necessarily have the same sign.

Let us call $ \mathcal {R} $ the concordance relation, defined on $ X $, the set of vertices or poles of $ R = G (X, U) $. According to Pacotte, it has the following properties:

(1) $ x \mathcal {R} y \ \textnormal {and} \ x \mathcal {R} z \Rightarrow $ dom($x) \equiv $ dom$ (y) $ beyond the concordance layer (xy);

(2) $ x \mathcal {R} y \ \textnormal {and} \ x \mathcal {R} z \Rightarrow $ dom$(x) \equiv $ dom$ (z) $ beyond the concordance layer (xz).

We deduce that dom($ y) \equiv $ dom($ z $), therefore that $ y \mathcal {R} z $ (transitivity of $ \mathcal {R} $).

Otherwise:

- If the layer ($ xz) \in $ dom($ x $), is beyond the layer $ (xy)  \in $ dom($ x $), then, $ y \mathcal {R} z $ according to ($ xZ$).

- Conversely, if $ y \mathcal {R} z $ according to ($ xz $), then the layer ($ xz) \in $ dom($ x $) is beyond the layer $ (xy) \in $ dom ($ x $).

In the intermediate case, the $ (xy) \cap (xz) $ concordance layer will be made of a part of $ (xy) $ and a part of $ (xz) $.

Note that the relation $ \mathcal {R} $ is not irreflexive, because otherwise, as it is transitive, there would be a {\it total order} on the set of poles of the network, which is not necessarily the case. Furthermore, the notion of concordance having little meaning for a single pole, Pacotte does not signal that it can be reflexive. However, if we admit that a pole is concordant with itself, which gives $ \mathcal {R} $ reflexivity, then the set of poles of the network becomes a {\it preorder}. On the other hand, $ \mathcal {R} $ is not antisymmetric in the weak sense:
\[
\forall x, y \in X, x \mathcal {R} y \ \textnormal{and} \ y \mathcal {R} x \nRightarrow x = y,
\]
The fact that two different poles $ x $ and $ y $ have a common layer in no way means that $ x = y $. A fortiori, $ \mathcal {R} $ is not antisymmetric in the strong sense, so that $ \mathcal {R} $ is not an {\it order} relation\footnote {That does not prevent Pacotte's networks to be antisymmetric graphs since, between two poles, the existence of an arrow oriented in one direction excludes the existence of an arrow oriented in the opposite direction (see below, § 5).}.

Pacotte also introduces the following elements: \\

\begin {dfn}
It will be said that two poles are in {\it perspectivity} when they are concordant and when their domains are tree-structured in the region which precedes the layer of concordance – which then becomes a {\it layer of perpectivity}. \\
\end {dfn}

\begin {prop} [Pacotte]
The relation of perspectivity is transitive.
\end {prop}

\begin {proof}
If $ x $ is in perspectivity with $ y $ and $ y $ with $ z $, then $ x $ is in perspectivity with $ z $. The transitivity of the concordance relationship generates the transitivity of the perspectivity relation. For the reflexivity and the antisymmetry, we can make the same remarks as above. \\
\end {proof}
 
 \begin {dfn} 
  We call {\it transitive graph} the 1-graph $ H (X, U) $ verifying:
 \[
 (x, y) \in U, (y, z) \in U \Rightarrow (x, z) \in U. \\
 \]
 \end {dfn}
 
\begin {prop}
 In a Pacotte network, the graphs $ H (C, W) $ and $ H (P, Z) $ associated respectively with the concordances and the perspectives are transitive graphs. \\
 \end {prop}

\begin {dfn}
We call {\it double layer} of a domain the system of connecting dipoles leading to a layer.
\end {dfn}

\subsection {Generalization}

Suppose we start from a positive pole. In this case, it is the direct current which generates the domain and the layer itself is called {\it positive}. The corresponding double layer is crossed from the inside to the outside. Arrived from the pole of the domain to the layer, instead of continuing, let's pass from the direct current to the reverse current. We cross backwards the dipoles of the double layer then generate the domains of all its negative poles. Hence the idea of {\it cusp domain}. \\

\begin {dfn}
We call {\it cusp domain} the domain traveled starting from the chosen pole (except the initial dipole) and turning back on a layer of the same sign. One thus generates a domain by layers of opposite sign and one ends finally in poles also opposite.
\end {dfn}

We can sometimes exclude, after turning back, at the end of the double layer, the return to the finish line. If only one channel remains, the cusp is then called a {\it reflection}.

According to the author, the preceding notions are enough to reconstitute the whole logico-mathematical domain and, in fact, the mathematician successively defines, in chapters 3 to 9 of his work, a "tree-like conception" ({\it conception  "rameuse" in French}) of the sum, of the ordered product, of combinations, of multidimensional spaces, of the logical product, of the network of predicates and of the network of the logic of propositions.

We are therefore faced with an extremely powerful approach which we cannot get rid of. Concerning in particular the ideas of {\it dominating a pole}, {\it layer of a domain} and {\it poles whose domains have a common layer}, Pacotte observes "that they must not be grasped as concepts implementing the abstract formal schemas of totality, class, definition by predicates or by conditions” (see \cite{Pac8}, 15). According to him, they in fact give the "right point of view" which then makes it possible to generate all formal logic and the most abstract mathematics.
                                                                                                                                                                                                                                                                                                                                                                                                                                                                                                                                                                                                                                                                                                                                                                                                                                                                                                                                                                                                                                                                                                                                                                                                                                                                                                                                                                                                                                                                                                                                                                                                                                                                                                                                                                                                                                                                                                                                                                                                                                                                                                                                                                                                                                                                                                                                                                                                                                                                                                                                                                                    
\section{Mathematical properties of Pacotte's networks}

As we have seen, a Pacotte tree network is a bi-chromatic graph. If we assume $ n $ vertices (this number depends on the network considered), it is a graph of order $ n $. Furthermore, the number of arcs going from a vertex $ x_{i} $ to a vertex $ x_{j} $ never exceeding 1, it is a 1-graph. This 1-graph is, moreover, {\it antisymmetric} because $ (x, y) \in U \Rightarrow (y, x) \notin U $.

\subsection {Connectivity}
A Pacotte network is also {\it connected}. Connectivity could be discussed since, in the representation of Pacotte (see Fig. 4), the dipole-units do not reach their ends. 

But it is in fact an artifice intended to distinguish them from the connecting dipoles\footnote{Pacotte notes that + or - signs placed on the poles would restore connectivity. If he does not use this solution, it is only because, graphically, these signs would have trouble to be well centered (see \cite{Pac8}, 8-9).}.

 In reality, the network is connected, which means that, for any pair $ (x, y) $ of distinct vertices, there exists a chain $ \mu [x, y] $ connecting these two points. In the case of a Pacotte network without circuits, the network is even in general {\it quasi-strongly connected} because one can verify, in particular when it is a tree or a rooted tree in the sense of graph theory, that, for any pair of vertices $ (x, y) $, there exists a vertex $ z (x, y) $ from which start both a path going to $ x $ and a path going to $ y $.

\subsection {Stability}

Now remember the following definition: \\

\begin{dfn}
A graph $ G $ is said to be {\it bipartite} if we can partition all of its vertices into two classes, so that two adjacent vertices belong to different classes.
\end{dfn}

A Pacotte network is obviously a bipartite graph. We will see that it is, moreover, composed of two stable sets. \\

\begin{dfn}
Let $ G (X, U) $ be a simple graph (i.e. an antisymmetric 1-graph). We say that a set $ S \subset X $ is {\it stable} if two distinct vertices of $ S $ are never adjacent.
\end {dfn}

A Pacotte network, i.e a bichromatic and bipartite graph (that is to say, having no cycle of odd length), necessarily admits a partition $ (A, B) $ of the set of its vertices into two stable sets $ A $ and $ B $. 

\subsection {Kernel. Grundy function}

\begin {dfn}
Given a 1-graph $ G = (X, \Gamma) $ we say that a set $ A \subset X $ is {\it absorbing} if, for all vertex $ x \notin A $, we have:
\[
\Gamma (x) \cap A \neq 0. \\
\]
\end{dfn}

\begin{dfn}
Given a 1-graph $ G = (X, \Gamma) $ we say that a set $ S \subset X $ is a {\it kernel} if $ S $ is both stable and absorbing. So we have:
\[
(1) \qquad \quad x \in S \Rightarrow \Gamma (x) \cap S = \emptyset \qquad \textnormal {(stable)};
\]
\[
(2) \ \quad x \notin S \Rightarrow \Gamma (x) \cap S \neq \emptyset \qquad \textnormal {(absorbing)}.
\]
\end {dfn}

When a Pacotte network is a circuit-free graph, we can therefore apply the following theorem. \\

\begin {thm} [\cite{Ber}]
If $ G = (X, \Gamma) $ is a circuit-free 1-graph, it admits a kernel; moreover, this kernel is unique.
\end {thm}

In this case, a Pacotte network has a unique kernel. We can also show that it accepts a Grundy function. Recall the definition: \\

\begin {dfn}
Let $ G = (X, \Gamma) $ be a circuit-free 1-graph. By definition, a function $ g (x) $ which associates with any integer $ x \in X $ an integer $ \ge 0 $, is a Grundy function  if $ g(x) $ is the smallest integer which does not appear in the set $ \{g (y)\ |\ y \in \Gamma (x) \} $.
\end {dfn}

It then follows from the previous theorem that: \\

\begin {thm} [Grundy 1939]
A graph without circuits admits a unique Grundy function $ g (x) $; at any point $ x $, this function $ g (x) $ is less than or equal to the length of the longest path from $ x $.
\end {thm}

A Pacotte network without circuits therefore also admits a Grundy function and only one.

Let us now suppose a Pacotte network admitting circuits (as in the case of Fig. 5 (left)). As these circuits are of even length (Theorem 4.3), the following theorem applies. \\

\begin {thm} [Richardson, 1953)]
If $ G = (X, U) $) is a 1-graph without odd length circuits, it admits a kernel (not necessarily unique).
\end {thm}

A Pacotte network which admits circuits therefore also admits at least one kernel. Consequently, as a graph without odd length circuits admits a Grundy function (see \cite{Ber}), it also admits a Grundy function.

\subsection {Couplings. Alternate chain}

It will now be recalled that the arcs or edges of the Pacotte network are, like the poles, of two kinds: rays of sheaf (unit dipoles) and opposite segments (connecting dipoles). We introduce here the following definition: \\

\begin {dfn}
Given a simple graph $ G (X, U) $, we call {\it coupling} a set $ U_ {0} $ of edges such that any two of the edges of $ U_{0} $ are non-adjacent.
\end {dfn}

For a given coupling $ U_{0} $, we generally draw the edges of $ U_ {0} $ by {\it thick} lines, the edges of $ V_ {0} = U - U_{0} $ by {\it fine} lines. \\

\begin {dfn}
We say that a vertex $ x $ is {\it saturated} by a coupling $ U_{0} $ if there exists an edge of $ U_ {0} $ attached to $ x $. In this case, we write $ S \in U_ {0} $. If $ S \notin U_ {0} $, we say that the vertex $ x $ is {\it unsaturated} by $ U_ {0} $.
\end {dfn}

A coupling which saturates all the vertices is called a {\it perfect coupling}. Perfect coupling is obviously maximum. \\

\begin {dfn}
We call {\it alternate chain} a simple chain (i.e. not using the same edge twice) whose edges are alternately in $ U_ {0} $ and in $ V_ {0} $ (i.e. alternately {\it thick} and {\it thin}). \\
\end {dfn}

\begin {thm} [Berge, 1957]
A coupling $ U_{0} $ is maximum if and only if there is no alternating chain connecting an unsaturated point to another unsaturated point.
\end {thm}

The Pacotte tree network presents alternating chains of two types of segments, the unit dipoles and the connecting dipoles, which can be represented respectively by thick and thin lines. The first will be in $ U_{0} $, the second in $ V_{0} $. In principle, all the points are saturated, these couplings are therefore both maximums and we have, in each, $ q = n / 2 $ edges.

\section{Summary of the first part}
	
	In summary of this first part, we can say that a Pacotte tree network is therefore a connected antisymmetric 1-graph (almost strongly connected if it is circuit-free), bichromatic and bipartite. It can admit cycles or circuits of even length and is therefore not necessarily a tree or a rooted tree in the sense of graph theory. It has no unsaturated vertex and we detect two stable subsets and alternating chains of two types of segments, which lead to the existence of two maximum couplings. It has at least one kernel and a Grundy function. The notions of positive and negative "poles", posed by Pacotte, are translated by the bicolouring of the vertices. The notions of "diffluence", "fluence", "confluence" and "domain" become conditions posed on cocycles. If the Pacotte network considered is circuit-free and provided with a root, a “rank” function can be defined on the poles, which makes it possible to identify the “normal layers” of the network as the corresponding levels of the graph $ G (X, U $). On the other hand, in any case, it is not possible to establish a total order on the layers. All that can be said is that the graphs associated with the concordances and perspectives of the poles are transitive graphs (possibly preorders if we admit the reflexivity of these relations).\\ \\

	\Large
\qquad\quad Partie 2 : The reconstruction of mathematics
\normalsize
\vspace{1\baselineskip}

Pacotte's goal in {\it The tree network} is not limited to the description of the network. The network is in fact only a means, not an end. In reality, faithful to the philosophical position of the author  - that of an "integral empiricism" (see \cite{Pac6}) - Pacotte has a much more ambitious goal: it involved reconstructing all of mathematics using this single notion of "tree network" and the related concepts\footnote{It should be noted that this project is part of a series of attempts which, at the time, were more or less in the same direction. We have already mentioned that of Brouwer, but we could speak of all the other attempts to link logic or geometry to the sensitive world - for example, that of Nicod (1920) or Tarski (1927) - projects which continued until to N. Goodman and J. Vuillemin. On all that, see our comments in \cite {Par}, 116-123.}. In an article which contains, so to speak, the "philosophy" of the work that we are commenting on, Pacotte explains it in these terms:

"In a recent work [this is {\it The Tree Network}], we traced the main lines of a theory of ramifications and we proposed to consider it, not as an application of formal logic and pure mathematics, but well as the foundation of both. Formal logic would thus cease to be an activity exercising on an indeterminate object: a well-characterized formal object would necessarily condition it” (see \cite{Pac13}, p. 46).

In this article, Pacotte explains that he caught a glimpse of this essential idea, then always had it more clearly conceived, during his previous studies concerning the axiomatic, the relationship of the formal to the real and the logic of action and expression, themes respectively encountered in his previous works (see \cite{Pac5}; \cite{Pac6}; \cite{Pac7}). But what philosophically founds such a project is undoubtedly the author's Bergsonian a priori, according to which time (or, let's say, becoming), that is to say a sort of ordered flow of facts or impressions, is our fundamental experience of the world, from which everything must be reconstructed (see also, on this subject, \cite{Pac11}, 164 and 178).

The major part of his text (pages 17-54) is therefore devoted to this project. This part of Pacotte's work, never commented on, is extremely abstract and difficult to follow. In particular, it suffers from a total absence of diagrams or illustrations which would have made it easier to read. In the following, we will limit ourselves to presenting the simplest of his reconstructions (the following diagrams are ours).

\section{The tree-like conception of the sum}
 
To establish his concept of a "branching sum" (chap. 3 of \cite{Pac8}), Pacotte considers the tree part of a domain and, in this part, the layers which constitute it. He introduces the following definitions: \\

\begin {dfn} [Perspectivity on layers and trees]
Two layers of the same tree are called {\it perspective layers}. Two trees with perspective poles are called {\it perspective trees}.
\end {dfn}
\vspace {1 \baselineskip}

\begin {dfn} [projectivity]
Two perspective layers or belonging to two perspective trees are said to be {\it projective}.
\end {dfn}
\vspace {1 \baselineskip}

\begin {thm} [Pacotte]
The perspectivity of the layers is not transitive.
\end {thm}

\begin {figure} [h]
\centering
\vspace {-2 \baselineskip}
\includegraphics [width = 4in] {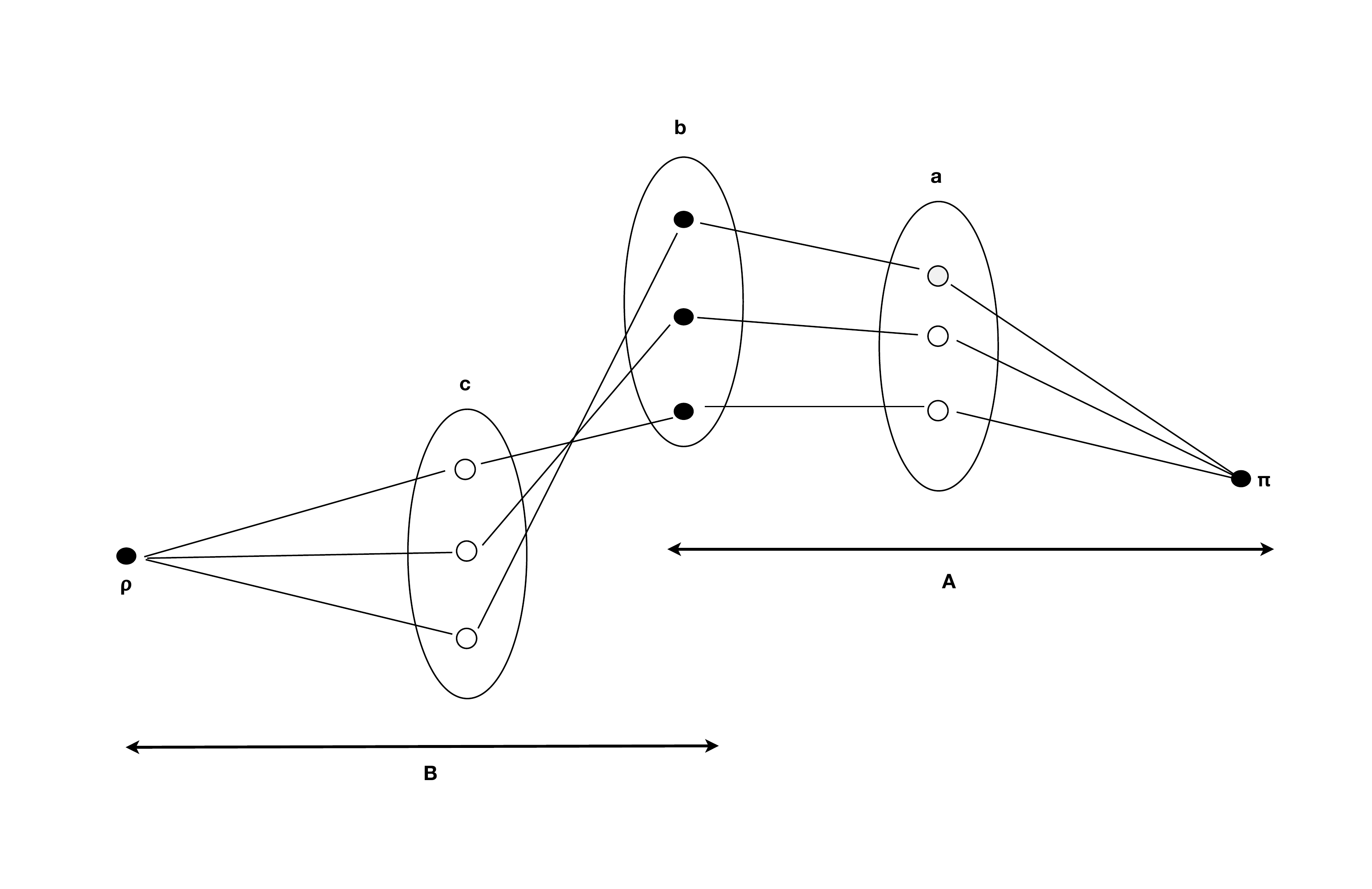}
\vspace {-2 \baselineskip}
\caption {Perspectivity of layers in a Pacotte network}
\label {fig: Cycle6}
\end {figure}

\begin{proof}
Let $ R $ be a Pacotte network, $ a, b, c $ layers of $ R $, $ A $ and $ B $ two trees of $ R $. Let $ \mathcal {P} $ be a relation of perspectivity defined on $ \mathcal{C} $, set of layers. Suppose $ a, b \in A $ and $ b, c \in B $ (see Fig. 7). In this case, we will have:

\[
a \mathcal {P} b \quad \textnormal {and} \quad b \mathcal {P} c \quad \textnormal {but not} \quad a \mathcal {P} c,
\]

since $ a $ and $ c $ are not layers of the same tree and belong to different poles. Therefore, unlike the perspective of the poles, the perspective of the layers is not transitive. However, the poles being in perspectivity (since the $ b $ layer is common to the two trees), the two layers will be {\it projective}.
\end {proof}

\begin {thm} [Pacotte]
The projectivity of the layers is a transitive relation.
\end {thm}

\begin {proof}
As we saw above, the projectivity of the layers is based on the perspectivity of the poles. As the perspectivity of the poles is a transitive relation, the projectivity of the layers is also a transitive relation.
\end {proof}

Now consider, in a tree structure $ A $, the region between a layer $ \alpha $ and the pole $ \pi $ of the domain $ D $ to which this layer belongs. In this region, let there be an arbitrary $ P $ pole, with the same sign as the tree structure. Pacotte explains that his domain $d$ {\it crosses} the layer, thus constituting a whole in itself. \\

\begin {dfn} [Tree section]
We will call {\it section} of the $ P $ tree this totality associated with the crossing of the layer $ \alpha $ by a domain $ d $ or, which amounts to the same thing, by the projection of the pole $ P $ on $ \alpha $ (see Fig. 8).
\end {dfn}

To understand, let us recall (Definition 4.1) that the domain of a pole is this region of the network which, from the pole in question, is traversed by a current of the same direction. The section is the sub-tree crossing the layer:

\begin {figure} [h]
\centering
\vspace {-2 \baselineskip}
\includegraphics [width = 4in] {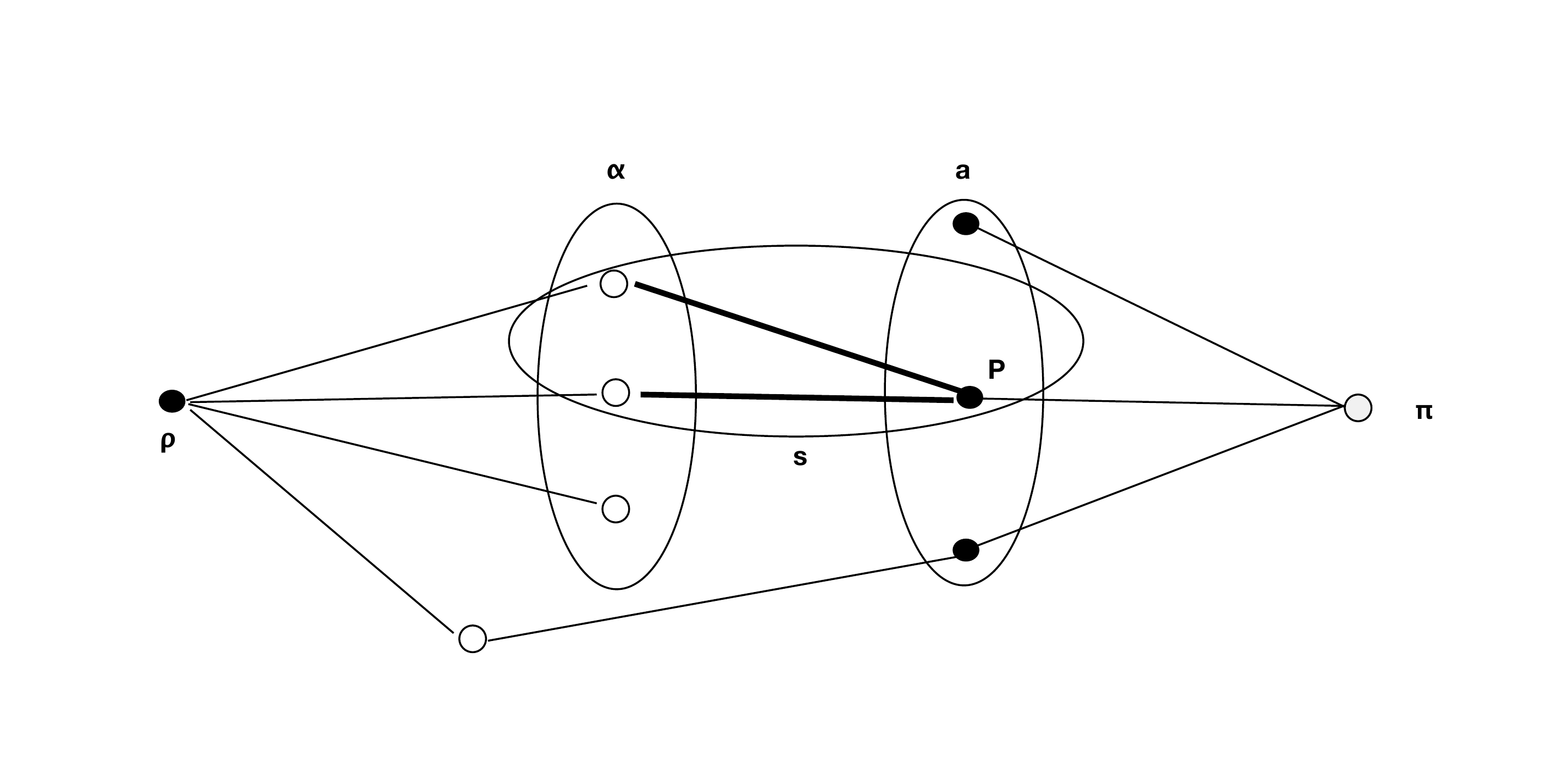}
\vspace {-2 \baselineskip}
\caption {Tree section in a Pacotte network}
\label {fig: Cycle7}
\end {figure}

Pacotte continues as follows: "Let us now take a layer $a$ in the region under consideration. Each of its poles, such as $ P $, projects onto $ \alpha $ in a layer. On the contrary, a pole of the layer $ - \alpha $ has for domain a simple chain, which crosses the layer $ a $ at a single point, such as $ -P $. This situation is explained by saying that there exists between the successive layers $ a $ and $ \alpha $ of a tree structure, a {\it multivalued correspondence}"(see \cite{Pac8} 17-18).

Recall, in fact, that a connecting dipole being the opposite segment to a sheaf radius, the system of connecting dipoles results in a layer which can be called an {\it associated layer}. In other words, if a layer is designated by $a$, the associated layer is designated by $-a$, regardless of the sign of $a$. Furthermore, if $P$ is a pole, then $-P$ is its associated pole on the associated layer, and this, whatever the sign of $P$. We can therefore represent things as in Fig. 9.

\begin {figure} [h]
\centering
\vspace {-2 \baselineskip}
\includegraphics [width = 4in] {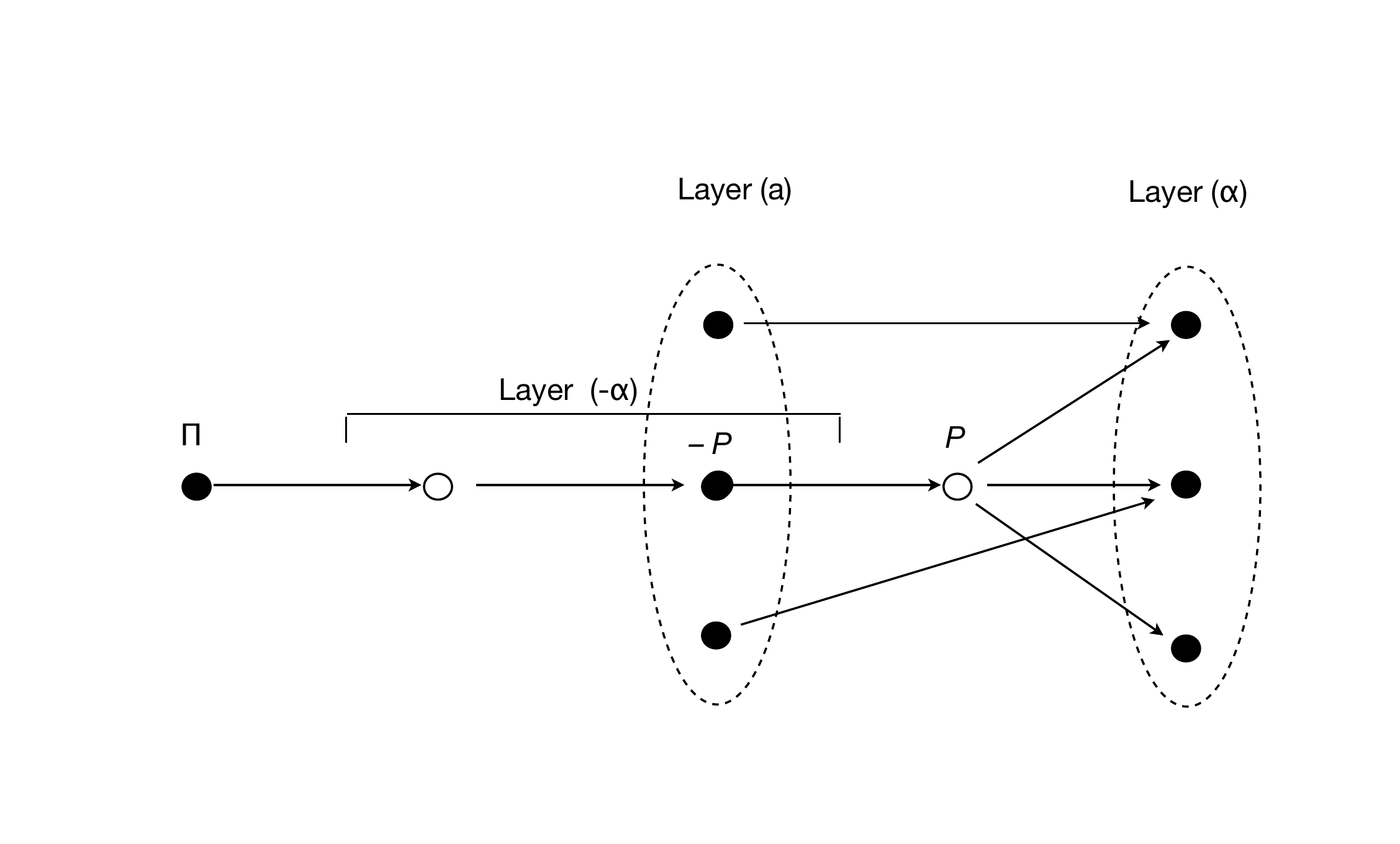}
\vspace {-2 \baselineskip}
\caption {Multivalued correspondence7}
\label {fig: Cycle8}
\end {figure}

Pacotte will then draw the notion of "arithmetic sum" from this conception.

“The application of this idea to the projective layers is remarkable. Such layers belong to perspective poles. Let us then consider a layer $ \alpha $ common to all the trees and located beyond all the layers considered (or at least not below). As each layer $ a $ is perspective to the layer $ \alpha $ and precedes it, there is between $ a $ and $ \alpha $ a one-to-one correspondence. Each layer $ a $ is then called a {\it sums of terms} (the terms are the constituent poles of the layer), the layer $ \alpha $ a set of object-units. The projectivity of the layers defines {\it the equality of sums}: all sums $ a $ are equal to each other and equal to the set $ \alpha $" (see \cite{Pac8}, 18).

In principle, there is no order on the terms of a sum. An order can however be established between the poles constituting a layer "when the tree presents only diffluences with two divergent currents, one of which does not diffuse, while the other diffuses until exhaustion" (see \cite{Pac8}, 18). Starting from the final bifurcation, the domain of an associated pole is an alternating chain, in fact a simple line that will be called $ L $, from which are starting lines going to the different poles of the layer. The order of the branches corresponds to the order of the poles on a layer. When we start from the final bifurcation to go back to the dominant pole, we reconstitute the tree in reverse order (if we compare this process to its genesis): it is very exactly, the operation of addition. “The order of addition is the reverse of the natural order of terms, which is defined by the final tree. A sum is a layer: to consider two orders of addition is to take into consideration two order trees, in perspectivity according to this layer\footnote{We could probably put this conception in relation to the pre-or post-fixed Polish notation, which translates, in effect, operations like addition into a tree structure (see, for example, \cite{Gon}, 114).}. The {\it commutativity} of addition is simply the possibility of taking a sum layer as a common layer of two order trees, establishing different orders” (see \cite{Pac8}, 19).

In modern terms, we will therefore say that Pacotte reduces the idea of arithmetic sum to two (for him) more fundamental notions which are: 

\begin{enumerate}
\item A transitive graph of projectivity whose elements are the layers.
\item A  multivalued mapping between pairs of layers. 
\end{enumerate}

We have already defined the notion of transitive graph (see Definition 4.11) and also that of multivalued mapping (see Definition 4.8), which is not very different from the general notion of {\it graph relation}. The transitive projectivity graph associated with the arithmetic sum as reconstructed by Pacotte is therefore a subgraph contained in the graph associated with the product of composition of multivalued applications.

It would remain, however, to explain in a more mathematical way the links between projectivity and sum. A projective algebra would presumably assume that we place ourselves in the framework of a (finite) projective geometry expressed in homogeneous coordinates, so that the operation of addition, associating elements of the finite with a point of projection located at finite distance, makes it possible to establish a true isomorphism between "projection" and "sum", which could then be read in a reversible manner. This isomorphism, easy to find, for example, in the case of a Boolean algebra and the famous 7-point projective Fano plane\footnote{On finite projective geometries, one can consult the classic work of \cite{Dem}. French Dominique Dubarle put it into practice in his interpretation of Hegel's logic, transforming the boolean algebra defined on a boolean set product $\{0, 1\}^2$ whose elements are isomorphic to 4 constants of structure (the 3 Hegelian constants ($U$, for universal, $P$ for particular, $S$ for singular, plus the null term $\Lambda$) into a projective algebra (see \cite{Dub}) defined on PG(2, 2). This one has a special addition $\oplus$, written in homogeneous coordinates. We get, for example : ($1, 0, 0) \oplus (1, 1, 0) = (0, 1, 0)$, which means : when we "add" (0, 0), identified with $\Lambda$ to (1, 0), identified with $U$ – terms which are finite representations –, we obtain the infinite representation (though located at a finite distance) of which they are the trace in the finite.  See also \cite{Hir1} and, for a for a more historical point of view, \cite{Hir2}, 34-81.}, would probably be more difficult to establish in any case. But it certainly exists. The fact remains that Pacotte's approach on this point is somewhat metaphorical: the correspondence is not really mathematically established between the idea of sum and that of "projection".

To finish with the representation of the sum in the Pacotte polarized tree network, one will note the difference of treatment between the properties of {\it commutativity} and {\it associativity}.  Commutativity does not belong to the transformations which one can make undergo a sum without in change its value. Associativity concerns the sum itself and is explained by Pacotte in terms of substitution and association. The ideas of belonging and inclusion are also explained by Pacotte in terms of membership of poles in layers, or membership of layers in other layers, the transitivity of these operations also founding, thereby, the laws of logical deduction (syllogisms and chains of syllogisms).

\section{The tree-like conception of the ordered product}

For understanding, this section must be accompanied by the following definitions: \\

\begin{dfn}[Parallel system]
"Let us call {\it parallel system}, writes Pacotte, the network generated by alternating chains, not presenting mutual confluences, and united by a sheaf making bridge in any place" (see \cite{Pac8}, 22).
\end{dfn}

A system of poles with the same sign (not necessarily that of the sheaf), where each chain is represented, is a {\it layer of the parallel system}. This layer is, at the same time, a layer of the domain of the center of the sheaf. \\

\begin{dfn}[Equal layers]
Two layers with the same sign of the parallel system will be called {\it prospectively equal} or, more simply, {\it equal}.
\end{dfn}
\vspace{1\baselineskip}

\begin{dfn}[Multiplication of layers]
“Let us consider parallel perspective systems (an example is given in Fig. 4, last picture) according to a positive layer $a$. Take one of these systems and choose a positive layer $a ’: a’ = a$. Suppose  now positive the sheaf of the parallel system (the sheaf is positive because its center $G$ is positive). In addition, suppose that this sheaf is thrown between $a$ and $a'$ so that the current flowing from $G$ goes to $a’$, not to $a$. Let us make the same assumptions and constructions for the other parallel systems: we thus define $a'', a'''$, ... Let us then assemble the centers $G$ of the sheaves of those parallel systems under a pole $B$, also positive (by connecting dipoles and sheaf). Call $b$ the layer of $G$, and $r$ the next layer of the domain of $B$, the one where the constituent poles of $b$ project along $ a', a''... $. Then, $b$ as the sum layer, represents the sum of $ b $ layers equal to $ a $, and the layer $ r $ as collection of units, the value of the sum. Thus, the network is that of the multiplication of $a$ by $b$" (see \cite{Pac8}, 23).
\end{dfn}

We try to represent this situation, difficult to imagine, in two successive diagrams. In the diagram of Fig. 10, $ C_{1}, C_{2} $ and $ C_{3} $ form alternating chains making up the "parallel system" ($ S $). They are joined by a "bridge" formed by the sheaf $ a $ (which is also a "layer" of the system, that is to say, as explained above, a system of poles with the same sign where each chain is represented). The layers $ a, a'$ and the intermediate layer $ i $ (positive sheaf with center $G$) are "perspectively equal", and therefore, this system - "parallel", since the layers are perspective - can be conceived as represented by a single layer (which can justify Pacotte's somewhat ambiguous expression: "as many systems as there are layers").

\begin{figure}[h] 
	   \centering
	      \vspace{-2\baselineskip}
	   \includegraphics[width=4in]{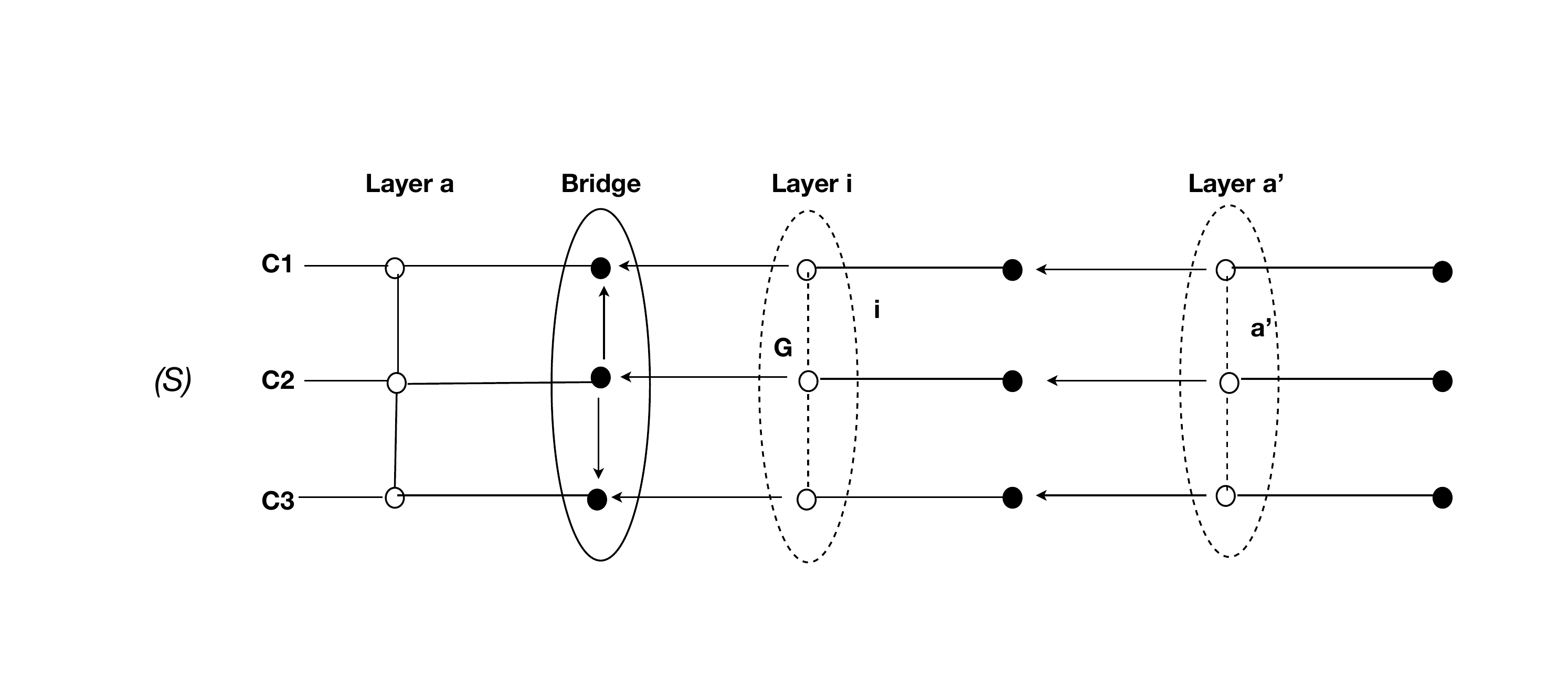} 
	      \vspace{-2\baselineskip}
	   \caption{The notion of {\it parallel system}}
	   \label{fig: Cycle9}
	\end{figure}

	We must now imagine (see Fig. 11) other systems $ (S'), (''), ... $ such as ($ S $), on which layers $ a'", a''', ... $ are defined with, each time, sheaves of center $ G_ {1}, G_{2}, ... $ Pacotte supposes that it is always possible to assemble these centers under a pole $ B $, also positive. Then the $ G_{i} $ themselves form a layer $ b = a^{(k)} $. Let $ r $ be the layer $ a^{(k + 1)} $ in the domain of $ B $. This is the layer where the constituent poles of $ b $ are projected, according to $ a' a'', ...$, etc. We thus arrive at the definition of the product: \\

\begin{dfn}
The sum layer $ b $ represents the sum of $ b $ layers equal to $ a $, and the layer $ r = a^{(k + 1)} $ - or collections of units - is the value of this sum . Thus, this network is that of the multiplication of $ a $ by $ b $.
\end {dfn}	

	\begin{figure}[h] 
	   \centering
	      \vspace{-1\baselineskip}
	   \includegraphics[width=4.5in]{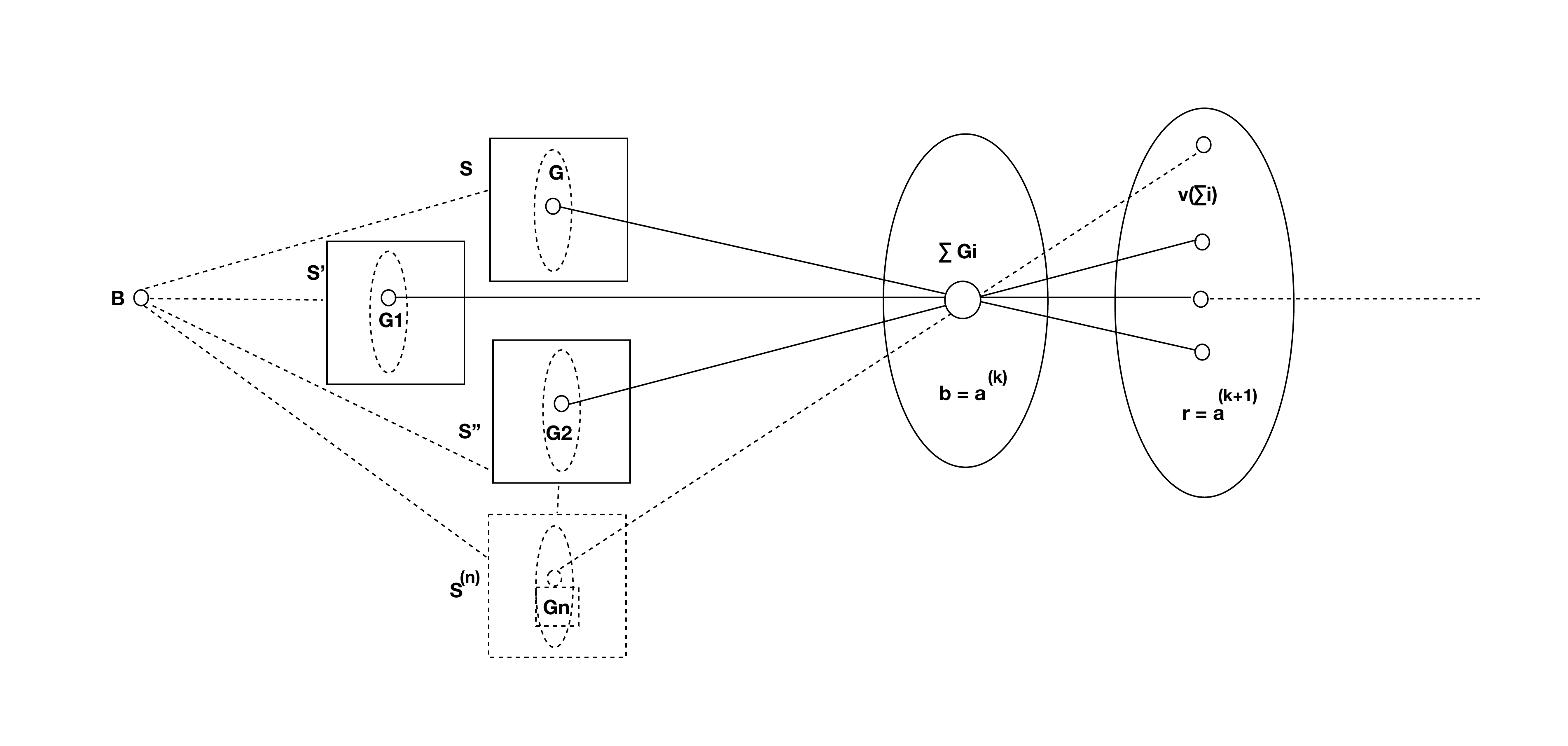} 
	      \vspace{-1\baselineskip}
	   \caption{The notion of {\it ordered product}}
	   \label{fig: Cycle10}
	\end{figure}

	For a theory favoring polarity, that is to say asymmetry or antisymmetry and relations of order, the questions of equivalence, symmetry and reciprocity are difficult to express. Pacotte performs wonders, in the following, to try to explain the commutativity of the product of two factors by introducing, in a way that is not really clear - to say the least - the concept of {\it reciprocal layer}.

But the introduction of a third factor removes the symmetry so that, at its bottom, the network of the ordered product of $ n $ factors $ a, b, c, .... $ does not manifest, in the general case, neither associativity, nor commutativity. Pacotte's statements here are quite unusual and very different from, say, Wedderburn-Etherington nonassociative products (see \cite{Par3}, 94-97). It must be understood on the one hand that the factors are assembled in a certain order, $ ab, abc, abcd ... $ (an indeterminacy subsisting in the order of the poles of the final layer), and that, on the other hand, the addition order and the tree order are reversed. In this context, we have:

\[
(ba) c = bac,
\]

because, as $ ba = ab $ (commutativity of the product of 2 factors), the assembly order is respected. On the other hand:
\[
c (ab) \neq cab,
\]
because the order here, even supposing the commutativity $ab =ba$, is not respected. In addition, the question of commutativity, in the case of a product of $ n $ factors, will not be resolved until later, with the examination of the “branching genesis of the multidimensional”.

Among the properties of multiplication, there is also the question of distributivity. How to express it in branching geometry? Pacotte gives himself a tree structure whose final layer is $ a $, as well as any layer $ \alpha $, so that the sum layer $ \alpha $ is equal to $ a $. Same construction with a second final layer $ b $ and any layer $ \beta $. $ \alpha_{\ mu} $ and $ \beta_{\ nu} $ then being the constituent poles of $ \alpha $ and $ \beta $, Pacotes names $\alpha(\mu) $ and $ \beta (\nu) $ their respective projections on $ \alpha $ and $ \beta $. Pacotte then successively forms:

\begin {itemize}
\item Parallel systems from $ a (\mu) $ and $ b (\mu) $;
\item Layers equal to these given layers, in sufficient number;
\item Products $ a (\mu) b (\nu) $ on these separate layers, for all the values of $ \mu $ and $ \nu $.
\end {itemize}

Pacotte groups the $ a (\mu) b (\nu) $ each time by means of a top-sheaf $ -A (\mu) B (\nu) $. The layer of vertices then becomes, as a sum, equal to the layer of $ -a (\mu) b (\nu) $. The layers $ A (\mu) B (\nu) $ and $ a (\mu) b (\nu) $ being reciprocal layers, respectively for $ \alpha $ and $ \beta $ and for $ a $ and $ b $, we can therefore simply say that we go from the $ \alpha \beta $ layer to the $ ab $ layer by substitution. We thus go from each pole $ \alpha (\mu)\beta (\nu) $ or $ A (\mu) B (\nu) $ to the product of the layers $ a (\mu) $ and $ b (\nu ) $. What has just been exposed constitutes the network of distributivity of the product in the case of 2 factors and a sharing of each factor. But - if we are not afraid of complications - we can obviously generalize the construction.

\section {The branching genesis of combinatorics, multilinear algebra and logic}

The rest of Pacotte's construction concerns combinatorics, multidimensional order and logic.	

\subsection{Combinatorics}

	How to generate, first, {\it combinatorial analysis} from a tree network? It is advisable to start from the ordered network of $ n $ factors, previously described. This network presents a {\it main tree} whose layer $ r_{p} $, of rank $ p $, represents the ordered product of $ p $ factors. Pacotte assumes that this tree has a vertex $ A $, identified with the top of the sheaf of layer $ a $. As we have a root and a sequence of normal layers, there is a total order on the layers. "The layer of rank 1 is the factor $ a $ writes Pacotte; that of rank 2, the product $ ab $; that of rank 3, the product $ abc $, and so on" (see \cite{Pac8}, 27). The previous layers obviously have their opposites ($ -a, -b, ..., -r_{p} $) and the products present the totality of the arrangements $ a_{k}, b_{j}, c_{i} ... $, where we take the first $ p $ letters. The order is made up from the order tree constituting the domain of each pole, from $ -r_{p} $ to $ -r_{2} $ (diffluence). Even for negative values of the layers, the arrangements will conventionally be designated by $ a_{i}, b_{j} ... $, etc. The network of the $ n $-th power of a layer $ a $ will play the role of network of arrangements $ a_{i}, a_{j}, a_{k} ... $, etc. The arrangements being relative to $ m $ types of objects, and each of them being an arrangement of $ p $ objects, these arrangements will be designated by $ A'(m, p) $.
		
	Let now be a pole $ N_{p} $ associated with $ M_{p} $ by a connecting dipole. Its domain is a tree structure contained in the main tree structure of the network. Let be a layer $ r_{p '} $ of rank $ p'> p $ and a system of poles $ N_{p '} $ associated with $ M_{p'} $. The totality of $ A'(m, p, p') $ will constitute the arrangements of $ m $ objects taken $ p'$ by $ p' $, starting with $ A'(m, p) $. "We will say that $ M_{p} $ represents the {\it intension} of $ A'(m, p) $ and $ N_{p} $ its {\it extension}, ideas to which will correspond, in the network of combinations (and especially in the network of predicates), the intension and extension of classical logic. With a few additional complications, Pacotte then succeeds in geometrically representing the logical product (formed of "diagonal" $ A'(m, p) $, i.e. whose entire domain merges into a single pole of the layer $ a $), as well as combinations of the type $ C(m, p) $. The network of $ C(m, p) $ then makes it possible to generate the {\it multidimensional} (see \cite{Pac8}, 32-33).

We cannot follow the reconstruction of Pacotte in detail. Perhaps it suffices to understand its meaning, which is very well summarized by the author at this point in its construction: for him, there is no doubt that the usual approach to founding mathematics consists - given the role of commutative multiplication in this discipline (which makes it possible to easily generate a multidimensional network of points aligned in $ n $ directions) - to take this very network as a foundation, and to give the formal logic in a completely independent way.
		
	“In opposition to this attitude, our doctrine, writes Pacotte, can be characterized as follows: starting from the branching space, demonstrating the branching genesis of the multidimensional and drawing from the branching space, particularly from the polarized network of $ C (m, p) $ and from the polarized multidimensional network, the logic of present or absent predicates, then the logic of propositions (see \cite{Pac8}, 34).

\subsection{Logical product multidimensional order}

In the {\it logical product}, for example, the $ n $ sets whose product is to be defined will be "the branching projections of $ n $ poles on a fundamental layer such as $ a $ for which we have constructed the network of $ m^n $ arrangements $ A'(m, n) $ and identified the system of $ m $ diagonal arrangements $ D (m, n) $, one for each pole constituting $ a $ "(see \cite{Pac8}, 36). Pacotte then endeavors to explain mathematically the properties of the logical product, in particular its distributivity, and how the extension of a logical product can be equal to the product of particular extensions of its components.

\subsection {The network of present and absent predicates}

In this context, what logicians call {\it predicate} is a constituting pole $ a_{i} $ of a layer $ a $ of the polarized network, on which is built the network of $ C(m, p) $, in other words the network of combinations. The extension to rank $ p'$ of an object assembling $ p $ predicates is therefore the extension $ \Xi_ {p'} C (m, p) $, i.e. the {\it class} of objects having $ p'$ predicates and presenting the $ p $ predicate of the object considered. The predicates are then the $ a_{i} $ whose presence or absence is asserted in a $ C(m, p) $, this assertion being necessarily thought in a multidimensional network whose each dimension presents only two elements, to which presence and absence correspond (see \cite{Pac8}, 41). So we have representations of this kind:
	
	\[
[(a_{1})_{i} (a_{2})_{j} (a_{3})_{k}...]_{p},
\]

with indices $ i, j, k $ which take only the values 0 and 1. We then simplify the notation by posing:
\[
a_{i} \ \textnormal{pour} \ (a_{i})_{1} ; \quad \dot{a}_{i} \ \textnormal{pour} \ (a_{i})_{0}.
\]
	
	It follows that the set of possible combinations, instead of being a power of 2, as in ordinary logic, is here a power of 3. In this geometry of logic, there is $ C_{m, p } $ plans with $ p $ dimensions, $ 2^p $ complexions of characteristic $ p $ for $ p $ determined dimensions, and $ 3^m $ total complexions – from where the formula:
	
	\[
\sum_{p=0}^m\ 2^p\ C_{m, p} = 3^m.
\]
	
	Pacotte then manages to make sense of the idea of a logical {\it proposition}. If the predicates are poles, the propositions are figures in space, corresponding to expressions obtained "by composing extensions by the operations of complementarity, product and logical sum" (see \cite{Pac8}, 45). Even if the formalism seems heavy, their usual meaning can be easily reconstituted. Depending on whether or not we take some of these complexions, we obtain this or that expression. Thus, the figures in the predicate space are represented by expressions of the type:
	
	\[
\overline{\Xi_{1}\Xi_{2}} \quad , \quad \overline{\Xi_{1}\overline{\Xi_{2}}}, ...
\]

where we have, for example:

\[
\Xi_{1} = \Xi_{m}(\dot{a}_{1}, a_{2}), \ \Xi_{2} = \Xi_{m}(a_{2}, \dot{a}_{3}),
\]

and so on. There will be expressions such that:\\

\begin{itemize}
\item  $\Xi_{1} + \overline{\Xi_{1}}  \qquad \qquad \textnormal{(sum of two complementary propositions)}$.
\item $\overline{\Xi}(a_{1}) = \Xi({\dot{a}_{1})} \quad \textnormal{(negation of}\ \Xi\ \textnormal{for}\ a_{1}\ \textnormal{true, i.e assertion of}\ \Xi\ \textnormal{for}\ a_{1}\ \textnormal{false})$.\\
\end{itemize}

The implication is more subtle to represent. If we have:

\[
\Lambda_{1} = \overline{\Xi_{1} \Xi_{2}}, \quad \Lambda_{2} = \overline{\Xi_{3} \overline{\Xi_{1}}},
\]
with:	

\[
\Xi_{1} = \Xi(a_{1}), \ \Xi_{2} = \Xi(a_{2}),\ \Xi_{3} = \Xi(a_{3}),
\]

to take $\Lambda_{1}$ is equivavent to exclude $\Xi_{1} \Xi_{2}$, so $\Xi_{1}(a_{1}, a_{2})$. In the same way, to take $\Lambda_{2}$ is to exclude:

\[
\Xi_{3} \overline{\Xi_{1}} = \Xi(a_{3}) \overline{\Xi(a_{1})} = \Xi(a_{3}, \dot{a}_{1}).\\
\]

It follows from this that we exclude complexions which present $ a_{3} $ without presenting $ a_{1} $. In other words, the presence of $ a_{3} $ entails or "implies" that of $ a_{1} $.

Now take the following new example:

\[
\Xi(a_{1} \dot{a}_{2}, \dot{a}_{3}) + \Xi(\dot{a}_{1} a_{2}, \dot{a}_{3}) + \Xi(a_{1} \dot{a}_{2} a_{3}),
\]

where the indices 1, 2, 3 are some values of the index $ i $ of $ a $. Complexions having no common elements, the logical product of their extensions is zero. The $ \Xi $ have no common complexion and we only retain the complexions which present one – and only one – of the three predicates. We express this situation in ordinary language by saying that we have either $ a_{1} $ or $ a_{2} $ or $ a_{3} $. The figure therefore expresses a disjunction of the predicates $ a_{1}, a_{2} $ and $ a_{3} $\footnote {see \cite{Pac8}, 45-46. This situation is somewhat reminiscent of the idea of a disjunctive normal form in usual formal logic.}.

Using figures expressing incompatibilities, Pacotte then comes to be able to geometrically represent the notion of {\it deduction}.

\subsection{The network of propositional logic}	

	The logic of propositions, which is the simplest in usual formal logic, is, from the point of view of the polarized network, the most difficult to reconstruct. Indeed, while the complete predicative propositions correspond to {\it determined} figures, the abstract propositions of the type $ p, q, ... $, here called rather $ \Phi_{1}, \Phi_{ 2}, ... $, correspond to {\it indeterminate} figures. Each of these figures is one of the $ 2^{(2^m)} $ possible figures, and we are therefore led to consider the variable figures $ \Phi_{1}, \Phi_{2}, ... $ as variables each taking $ 2^{(2^m)}$ values. “A proposition therefore admits a latitude, no longer in the space at $ 2^m $ points (complexions of predicates present or absent), but in a discrete space $ E $ whose number of points is expressed by:
\[
|E| = (2^{(2^m)})^N.
\]
	
	To form the network of this space, the starting point is necessarily the network of present and absent predicates. We then replace the singular, individualized and concrete figures $ \Phi_{1}, \Phi_{2}, ... $ with variable figures or complexions. Knowing that there will be as many layers in the network as there are $ \Phi $ variables, that is $ N $, the network of the commutative product of the $ N $ layers brings up expressions whose maximum form will be:	
\[
\Phi_{1}(\sigma_{1}), \Phi_{2}(\sigma_{2}), ..., \Phi_{N}(\sigma_{N}),
\]

where the index $ \sigma $ varies, for each $ \Phi $ separately, from 0 to $ 2^{2^m} $.

Then, if we have:

\[
\Phi_{1}(\sigma'_{1}). \Phi_{2}(\sigma'_{2}) = 0,
\]

we will also havei $\Phi_{1} \Phi_{2} = 0$ for a system:

\[
\Phi_{1}(\sigma'_{1}), \Phi_{2}(\sigma'_{2}), ..., \Phi_{N}(\sigma'_{N}),
\]

and for all systems conserving these values  $\sigma'_{1}$ et $\sigma'_{2}$, i.e. for all systems:

\[
\Xi_{N}\ [\Phi_{1}(\sigma'_{1}), \Phi_{2}(\sigma'_{2})],
\]

where $ \Xi $ denotes an extension in the space of variables $ \Phi $. Conversely, for any system of values of the two variables $ \Phi_{1} $ and $ \Phi_{2} $, we have a particular $ \Xi_{N} $. All these $ \Xi_{N} $ being mutually external to each other (at least one coordinate varies from one to the other), their sum is the figure:
\[
R(\Phi_{1}\ \Phi_{2}) = 0.
\]

Simply put, at each point in $ R, \ \Phi_{1} \ \Phi_{2} $ is zero, while its complement $ \overline{R} (\Phi_{1} \ \Phi_{2}) $ is not zero.

Propositions of the type $ \Phi _{\alpha} $ is in $ \Phi _{\beta} $ will be put in the form $ \Phi_{1} \ \Phi_{2} = 0 $ (which will be called propositions of the type $ p $) and the integer space $ E_{\Phi} $ of $ N $ variables $ \Phi $ will admit expressions such as:

\begin{align}
R(p_{1}) + \overline{R}(p_{1}) = R(p_{2}) + \overline{R}(p_{2}) = ...\\
= [R(p_{1}) + \overline{R}(p_{1})][R(p_{2}) + \overline{R}(p_{2})] = ...\\
= \sum R(p_{1})\ \overline{R}(p_{1})\ \overline{R}(p_{3})\ R(p_{4}) = ...
\end{align}

where $p_{1}, p_{2}...$ are propositions of the type $p$.

These expressions obviously recall the expressions of the total space of the predicates complexions.

Finally, let us now give two examples of figures representing known propositions:

The figure:

\[
R_{1} \overline{R}_{2}\overline{R}_{3} + \overline{R}_{1} R_{2} R_{3} + \overline{R}_{1} \overline{R}_{2} R_{3},
\]

where $ R $ denotes $ R(p_{i}) $, corresponds to a proposition which, expressed in ordinary language, is of the type: "either $ p_{1} $ is true, or else $ p_{2} $ is true, or $ p_{3} $ is true".

Another example is the following: the implication of the type "$ p_ {1} $ implies $ p_ {2} $" corresponds to the geometrical situation of the polarized network: $ R_{1} $ is a part of $ R_{2} $, relation which is not expressed by a region of $ R $ but by a relation between two figures. A disjunction and an implication are therefore, from the point of view of the polarized network, very different propositions\footnote {Instead of having, as we know, in ordinary logic, a definition of implication in terms of disjunction, since $ p \supset q = \neg p \vee q $.}. But the syllogism works, since, due to the transitivity of the inclusion relation, if $ R_{1} $ is a part of $ R_{2} $ and $ R_{2} $ is a part of $ R_{3} $, then $ R_{1} $ is a part of $ R_{3} $, which means that $ p_{1} $ implies $ p_{3} $.

Of course, the negation of a proposition $ p $, in other words $ \Phi_{1} \Phi_{2} \neq 0 $ ($ \Phi_{1} $ and $ \Phi_{2} $ have a common part ) and the disjunction of several propositions $ p_{1}, p_{2} ... $ are also expressed by a figure $ R $. The network therefore makes it possible to process all the propositions of classical calculus (for examples, see \cite{Pac8}, 53-54).

\section{Conclusion}

	A reviewer of Pacotte's book on {\it La physique théorique nouvelle} (1921) noted the solidity of this one. He also stated that he did not see where the book lacked mathematical equipment, even though in most cases the author used the mathematical equations "written with words". For those interested, he observed, however, that "the equations written in symbols" would have made reading easier (see \cite{Rev}, 485). We could make the same remarks about {\it Le Réseau arborescent}. Except that writing in words must have been a theoretical choice by Julien Pacotte. Theoretical and not just educational: for him, mathematics had to be based on something other than themselves. Now, what is closer to the empiric world than natural language? In this sense, our attempt to translate Pacotte's language into mathematical language could pass for a misinterpretation or, in any case, an infidelity, with regard to the author's project. At the same time, it was necessary to explain well. And when the subject becomes complex, the choice to write in natural language becomes difficult to hold. In {\it Le Réseau arborescent}, we see the symbols reappear towards the end of the book. We then regret that they did not return earlier: this would have clarified the discourse and removed some ambiguities.
	
	Despite our efforts to understand, reformulate mathematically, illustrate diagrams the text of Pacotte, we do not hide that it remains sometimes a little enigmatic, like, by the way, its author who, apart from his books, does not will hardly have left a trace behind him. All the more reason, obviously, to try to get him out of oblivion. It was a great attempt, in fact, to start from the Bergsonian becoming, to formalize this concept by the notion of "tree network", then, the properties of the network being described, try to reconstruct from there by branching the whole of mathematical thinking - that is, ultimately, the form of all that is. Pierre de la Ramée himself did not go that far\footnote{Note also that the tree, limited to a principle of specification, is already a legacy in Ramus. First, it is from the XI $^{th}$ century that the motif of the tree is deployed in Western culture, especially in religious art, through the tree of Jesse, which represents the genealogy of the Christ. It then spreads quickly and appears in illuminations, frescoes, stained glass, sculptures, and begins to symbolize with abstract trees from logical reflection. As Jean Lecointe writes, "the end of the Middle Ages and the beginning of the Renaissance made the tree more and more proliferating. The oldest is undoubtedly the Porphyry tree, of the IV$^{th}$ century, which represents in this form the subdivisions of the genera of being, from indeterminate being to rational being. For the tree, both the logical tree and that of Jesse, is the order of the {\it specification}; the report is the same, the ancient logicians never cease to emphasize it, from race to individual, and from genus to species, to use Aristotelian terminology ”(see \cite{Lec}, 27-28).The Porphyry tree therefore has abundant posterity. We find it in Rodolphe Agricola, at the end of the XV$^{th}$, to unify logic and rhetoric, and Erasmus and Ramus will basically only popularize the use. During this time, on the vaults of the monuments, the flamboyant art deploys the same movement of ramification, starting from the elements of division of space, still elementary, of the first Gothic with intersecting ribs infinitely multiplied by ribs or counter-curves and finials. "The unbridled taste that Pierre de la Ramée, says Ramus, inherits from his predecessors for the tree ramifications, so in harmony with his name, must undoubtedly be related to his recognition of the individual as {\ it species} to full part”, concludes J. Lecointe (see \cite{Lec}, 28). See also \cite{Dum}).}

And, up to the logician Léo Apostel (see\cite{Apo}, 157-230), then to the work that we have been able to carry out ourselves recently on the theory of classifications, no one had launched such a daring project.
	
	If Pacotte's essay is not fully convincing, it is not only because mathematics has grown so much since his time that the task of bringing sophisticated structures back to the shape of a tree network could prove to be today difficult, if not impossible, to achieve. It is also because the means available to the author for this reconstruction were probably a little too limited: Pacotte writes at a time when the theory of graphs is still on the verge of being born. It was in 1936, the year when {\it Le Réseau Arborescent} appeared, that the Hungarian mathematician Dénes König published his {\it Theorie der endlichen und unendlichen Graphen} (see \cite{Kon}) and it was only in the post-war years that the theory of graphs (with Tutte, Berge and others) as well as the theory of transportation networks (with Ford and Fulkerson) will really develop. The use of projective geometry, for its part, also remains, in all of Pacotte's text, somewhat metaphorical and imprecise. One could wonder why. That projective geometry could found the notions of sum and algebraic product was however an idea already present in von Staudt from 1856-57, notably in the second volume of {\it Beitrage} (see \cite{Sta}). 	
	
	In addition, the Italian mathematician Gino Fano (see \cite{Fan}) had proposed an axiomatic treatment of projective finite geometry in 1892. One can also mention, on the same subject, the remarkable article by Oswald Veblen from 1906 (see \cite{Veb}). But it is true that the use of these formal tools contravened Pacotte's project, one of the last attempts to suspend mathematics from what was not it. A weakened (and therefore more reasonable) version of this project – to base mathematics on constructive approaches – which will develop, with the success we know, via Brouwer and Heyting, would undoubtedly help to definitively bury his attempt. It was all the same interesting to show what was, in its essence, this essay, even if it passed for the witness of a bygone era.

{}
\end{document}